\documentclass[10pt,journal,final]{IEEEtran}
\usepackage[comma,numbers,square,sort&compress]{natbib}
\usepackage{amsmath}
\usepackage{amssymb}
\usepackage{float}
\usepackage{amsthm}
\usepackage{graphicx}
\usepackage{subfigure}
\usepackage{epstopdf}
\usepackage{color}
\usepackage{verbatim}
\usepackage{tikz,times}
\usepackage{algorithm}
\usepackage{algorithmic}

\DeclareMathOperator{\blockdiag}{blockdiag}
\DeclareMathOperator{\diag}{diag}
\DeclareMathOperator{\Cov}{Cov}
\DeclareMathOperator{\Tr}{Tr}

\theoremstyle{plain}
\newtheorem{theorem}{Theorem}[section]
\newtheorem{proposition}{\textbf{Proposition}}[section]
\newtheorem{lemma}{\textbf{Lemma}}[section]

\newtheorem{definition}{\textbf{Definition}}[section]
\newtheorem{assumption}{\textbf{Assumption}}[section]

\theoremstyle{remark}

\allowdisplaybreaks[4]
\begin{document}

\title{Distributed Design of Robust Kalman Filters over Corrupted Channels}

\author{
	Xingkang~He,  
	Karl H. Johansson,  
 Haitao Fang	
	\thanks{The work is supported by the Knut \& Alice Wallenberg Foundation, the Swedish Research Council, and 	
		the National Natural Science Foundation of China (Grant No. 11931018)}
	\thanks{X. He and K. H. Johansson are with the Division of Decision and Control Systems, School of Electrical Engineering and Computer Science. KTH Royal Institute of Technology, SE-100 44 Stockholm, Sweden (xingkang@kth.se, kallej@kth.se)}
	\thanks{H. Fang is with LSC, NCMIS, Academy of Mathematics and Systems Science,
		Chinese Academy of Sciences, Beijing 100190, China; He is also with School of Mathematical Sciences, University of Chinese Academy of Sciences, Beijing 100049, China (htfang@iss.ac.cn)}
}

\maketitle

{ 
	\begin{abstract}
		We study  distributed filtering for a class of uncertain  systems over corrupted communication channels.
		We  propose a distributed robust Kalman filter with stochastic gains, through which upper bounds of  the conditional mean square estimation errors are calculated online.  We  present a robust collective observability  condition, under which the mean square error of the distributed    filter is proved to be uniformly upper bounded if the  network is strongly connected.  For better performance, we modify the  filer by introducing a switching fusion scheme based on a sliding window. It provides a smaller upper bound of the conditional mean square error.  
		Numerical simulations are provided to  validate the  theoretical results and show that the filter scales to large networks.
		\end{abstract}
		
	}



\section{Introduction}
In recent years, networked state estimation problems for sensor networks are drawing more and more attention due to their many applications \cite{rao1991fully,yu2016distributed,saucan2017multisensor}.
{   Compared to the centralized methods,  distributed algorithms, implemented at each sensor, are more resilient to network vulnerabilities, require less energy-consuming communication, and are able to perform parallel processing.
} Thus, a growing number of  researchers are focusing on the study of distributed state estimation problems \cite{Olfati2007Distributed,olfati2009kalman,khan2008distributing,kar2011gossip,gupta2015error}.
{ System uncertainties and communication imperfections pose, however, great challenges to the implementation and use of  existing distributed filters.  Thus, it is important to study distributed  robust filters for   real-time state estimation of uncertain systems.}



    { System uncertainties exist in most applications in both the dynamics and measurements.}
Multiplicative noise arises in many situations \cite{tuzlukov2002signal}.  
When  system dynamics suffer  multiplicative noise, it is challenging to design effective filters due to the state-dependent uncertainty. 
The authors in \cite{yang2002robust} studied   centralized estimation problems for systems with  multiplicative noise and parameter uncertainties.
In  \cite{feng2013distributed},  distributed fusion estimation   for   systems with multiplicative and correlated noise was studied. In \cite{liu2016minimum}, the authors studied  distributed filtering  for  systems with multiplicative noise in the dynamics when the network is given by a complete graph. 
      Measurement degradation  usually comes from sensor or communication limitations \cite{liu2016minimum,wen2016recursive,Yang2014Stochastic}.
    A detailed study on Kalman filters with measurement degradations was given in \cite{dey2009kalman}.
     In \cite{wen2016recursive}, a distributed filter was proposed for a state-saturated system with degraded measurements and quantization effects.
      A robust estimation problem based on   randomly dropped measurements was studied in  \cite{zhou2016robust}.
     A distributed robust filter   was  provided in \cite{ugrinovskii2011distributed} for a class of linear systems with uncertain measurements.     Moreover, to deal with
     		random changes in model   structures and parameters in the real systems, 	
     		some  robust  filtering  approaches were proposed for  systems with unknown parameters under non-Gaussian measurement noise  \cite{prsicrobust,stojanovic2016joint,stojanovic2020state}	and for nonlinear uncertain Markov  jump systems  \cite{yin2014robust,dong2020robust}. { Most of the above results were studied in a centralized framework, and
     	for the distributed algorithms, few connections between  filter performance and system uncertainties were provided.}

In the  literature of distributed estimation over sensor networks \cite{Cat2010Diffusion,stankovic2009consensus,Yang2014Stochastic,liu2016minimum,li2017distributed,wen2016recursive,He2017Consistent,he2019distributed,he2020distributed}, { a common assumption is that the communications between sensors are noise-free. This is,  however, difficult to fulfill in practice \cite{kar2008distributed}}. Uncertainty induced by channel noise makes it more challenging to design and analyze distributed filters.
The authors in \cite{Khan2014Collaborative} investigated the design of distributed filters with constant filtering gains and fusion weights, and gave  conditions to ensure the boundedness of the mean square error (MSE).
 In \cite{Cat2010Diffusion}, a   distributed filter was proposed  by combining a diffusion step with the Kalman filter. 
The filter performance  was analyzed under the assumption that  each sub-system  is observable, which    is a restrictive condition for  high-dimensional systems.
Time-varying  distributed filters can  achieve better performance than static \cite{speranzon2008distributed,Battistelli2014Kullback, Wang2017On}. However, 
authors of \cite{speranzon2008distributed,Battistelli2014Kullback, Wang2017On}   all assumed  perfect communication.
Although \cite{ji2017distributed} studied the case that the state estimates suffer channel noise,   the  parameter matrices were required to be perfectly transmitted. 
{  The design of  distributed robust filters exposed to corrupted communication channels   needs further investigation.}

The main contributions of this paper are summarized in the following.
{ \begin{itemize}
		\item  For  systems suffering multiplicative stable noise   and   measurements exposed to fading and additive noise, we design a robust distributed Kalman filter  able to handle corrupted communication channels (Algorithm \ref{alg:A}). The filter is shown to be conditionally consistent in the sense that the MSE is conditionally bounded.

		\item We  extend   traditional collective observability to   robust collective observability, under which the MSE of the distributed robust Kalman filter is proved to be uniformly upper bounded for any strongly connected network (Theorem \ref{theorem_boundedness}). 
		
		\item We modify the proposed distributed robust Kalman filer by introducing a switching fusion scheme based on a sliding window  and past state estimates (Algorithm \ref{alg:B}). Adaptive covariance intersection (CI) weights are obtained by solving semi-definite programming (SDP) problems at the preset intervals. It is proved that the modified filter inherits the main properties of the distributed robust Kalman filter (Theorem \ref{theorem_boundedness2}),  but in addition	 provides a smaller upper bound of the conditional MSE.  
\end{itemize}}

This paper presents  significant contributions compared to the existing literature. In particular, first, compared to   \cite{Cat2010Diffusion,stankovic2009consensus,Yang2014Stochastic,liu2016minimum,li2017distributed,wen2016recursive,He2017Consistent,he2019distributed} where the communications are required to be noise-free, or \cite{ji2017distributed} where  the transmitted state estimates suffer channel noise, this paper studies a more general case of channel corruption. We allow that both the transmitted estimates and parameter matrices can be polluted by channel noise. 
Second,  this paper does not make the assumption   that the nominal systems have to be stable \cite{liu2016minimum,li2017distributed,wen2016recursive} or that each sub-system  is observable \cite{Cat2010Diffusion,stankovic2009consensus,Yang2014Stochastic}.
{  Moreover, different from \cite{liu2016minimum,li2017distributed,wen2016recursive},  the design of the filters in this paper is based on the  information from the local sensor and the neighbor communications.}
Third, compared with the existing results \cite{ji2017distributed,liu2016minimum,li2017distributed,wen2016recursive,He2017Consistent,he2019distributed}, using the neighbor estimates in a sliding window, the switching fusion scheme of this paper can utilize the state estimates more efficiently.

The remainder of this paper is organized as follows: Section~\uppercase\expandafter{2} presents the problem formulation. The filter design and performance analysis are  given in Section~\uppercase\expandafter{3}.  Section~\uppercase\expandafter{4} provides the modified filter based on
a sliding-window  method. After Section~\uppercase\expandafter{5} gives  numerical simulations,  Section~\uppercase\expandafter{6} concludes this paper.

\subsection*{Notations}
 Superscript $T$ represents transpose. The notation $A\geq B$ ($A>B$), where $A$ and $B$ are real symmetric matrices, means that $A-B$ is a positive semidefinite (positive definite) matrix.
 We denote $\textbf{1}_n$   an $n$-dimensional vector with all elements   one,   $I_{n}$   the identity matrix with $n$ rows and  columns,  
$\mathbb{R}^n$   the set of $n$-dimensional real vectors, and $\mathbb{N}$   the set of natural numbers. The operator $E\{x\}$ denotes the mathematical expectation of the stochastic vector $x$,  and $\Cov\{x\}=E\{(x-E\{x\})(x-E\{x\})^T\}$. We use $\blockdiag\{\cdot\}$ and $\diag\{\cdot\}$ to represent the diagonalizations of  square matrix elements and scalar elements, respectively.  The trace of matrix $P$ is denoted by $\Tr(P)$. For a real-valued matrix $A$,
$\rho(A)$ denotes the spectral radius and  $\|A\|_2=\sqrt{\rho(A^TA)}$. The scalar $\lambda_{\max}(B)$ is the maximal eigenvalue of the real-valued symmetric matrix $B$, and  $\sigma(\cdot)$ is the minimal $\sigma$-algebra operator generated by   a collection of subsets.
{For reading convenience, main symbols of this paper are provided in Table \ref{table_notations}.}

\begin{table*}[htbp]
	{{ 
		\caption{Main symbols in this paper: $k$ and $m$ stand for time instants, $i$ and $j$ stand for sensor labels}  
		\label{table_notations}  
		\scalebox{1}{
			{\renewcommand{\arraystretch}{1.5} 
				\begin{tabular}{|c|c|c|c|c|c|c|c|}
					\hline  
					symbol   & meaning&symbol &meaning& symbol&meaning&symbol   & meaning\\
					\hline 
					
					$x_{k}$ & state& $w_k$&process noise& $y_{k,i}$&measurement& $v_{k,i}$&measurement noise\\		\hline
					
					$\epsilon_k$& multip. noise&
					 $\gamma_{k,i}$ &fading factor&$N$ &sensor number&$A_{k}$ &system matrix\\		\hline
					 
					 $C_{k,i}$ &measurement matrix &$F_{k}$ &multip. noise matrix&$\mathcal{V}$ &node set&$\mathcal{E}$ &link set\\		\hline

			$\mathcal{V}=[a_{i,j}]$ &adjacency  matrix&$\mathcal{N}_i$ &neighbor set &$\varepsilon_{k,i,j}$ &channel noise&$\mathcal{D}_{k,i,j}$ &channel noise \\\hline

						$Q_{k}$ &cov. bound of $w_k$ &$\mu_{k}$ &cov. bound of $\epsilon_k$ &$\varphi_{k,i}$ &cov. bound of $\gamma_{k,i}$ &$ R_{k,i}$ &cov. bound of $v_{k,i}$ \\\hline
						
$\varUpsilon_{i,j}$ &bound of $\varepsilon_{k,i,j}$&$\mathcal{D}_{i,j}$ &bound of $\mathcal{D}_{k,i,j}$ & $P_{0}$ &bound of $E\{x_0x_0^T\}$& $\tau_{k,i}$ & mean of $\gamma_{k,i}$\\\hline

$\hat x_{k,i}$ &fused estimate&$\bar x_{k,i}$ &predicted estimate 
&
$\tilde x_{k,i}$ &updated estimate&
$\hat{\tilde x}_{k,i,j}$ &corrupted $\tilde x_{k,i}$\\\hline

$\mathcal{W}_{k}$ &channel noise $\sigma$-algebra
 &$K_{k,i}$ &filter gain
&
$W_{k,i,j}$ &fusion weight &$\varPi_{k}$ &bound of $ E\{x_{k}x_{k}^T\}$\\\hline

 $\bar P_{k,i}$ &parameter  in prediction& $\tilde P_{k,i}$ &parameter  in update 
&$P_{k,i}$ &parameter  in fusion&$
\bar{\tilde P}_{k,i,j}$ &corrupted $
\tilde P_{k,j}$\\\hline

 $\Delta_{i}$ &optimization interval&$L$ &window length
&
$\Phi_{k,m}$ &transition matrix& $\bar N$ &observability parameter \\\hline

				\end{tabular}
			}
		}
	}}
\end{table*}

\section{Problem Formulation}
This section presents a motivating example followed by some preliminaries together with the problem formulation.

\subsection{Motivating example}
\begin{figure}[htp] \centering
	\subfigure[Evolution of the field] { \label{fig:observe}
		\includegraphics[scale=0.3]{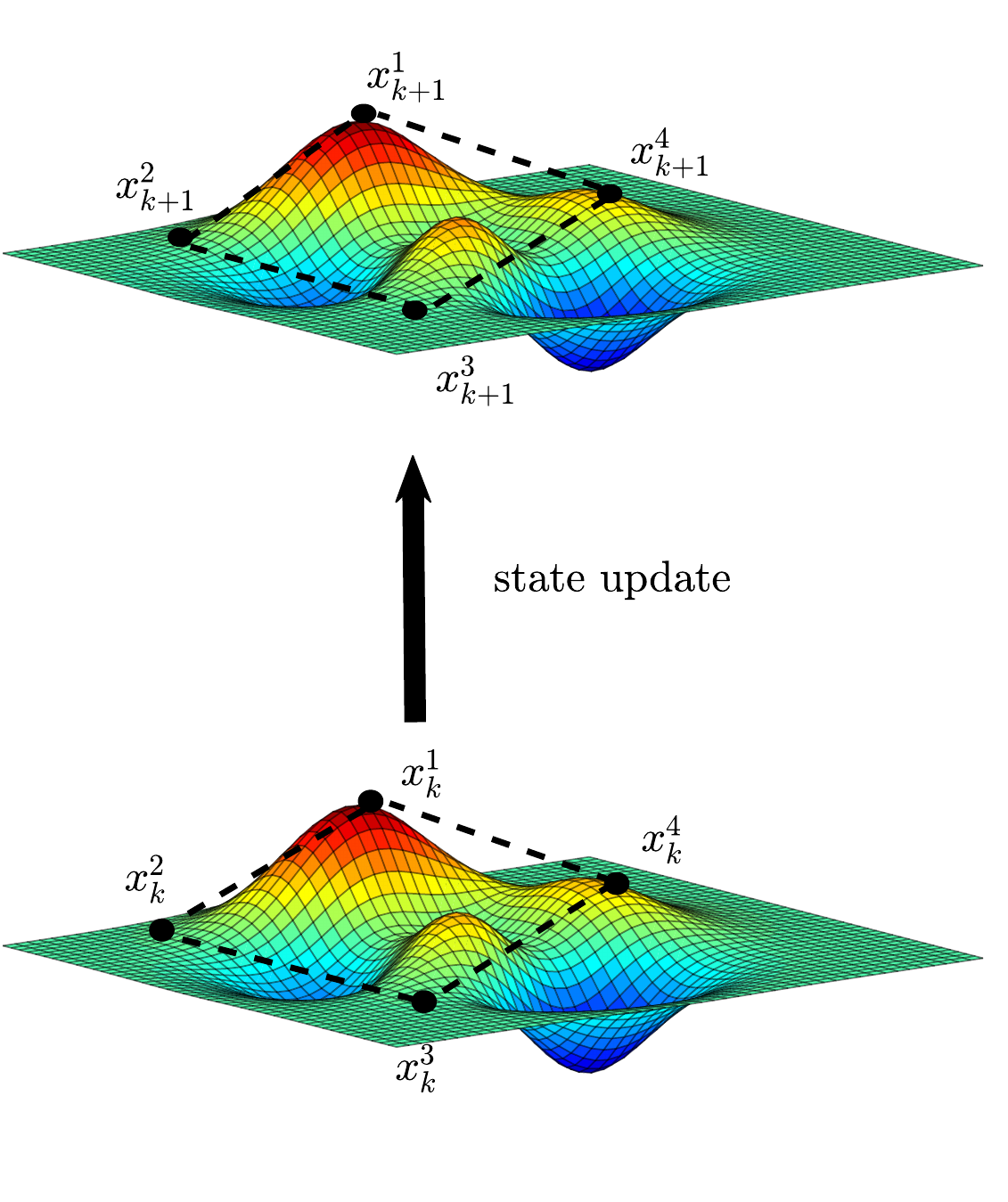}
	}
	\subfigure[Distributed sensing and estimation] { \label{fig:communicate}
		\includegraphics[scale=0.35]{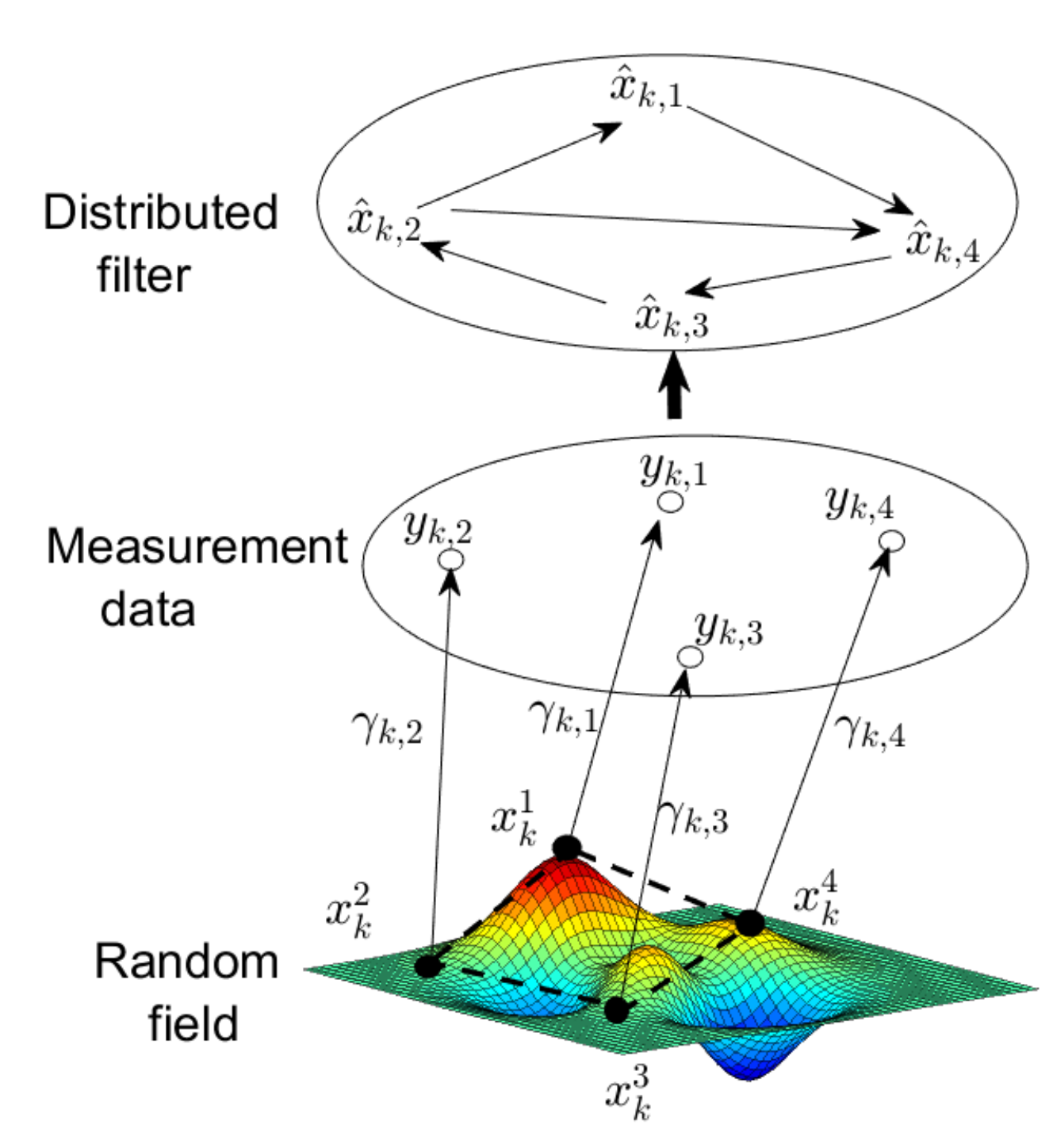}
	}
	\caption{ A random temperature field  over a geographical area. The evolution of  the  field  is driven by some stochastic process $w_{k}$. The right figure illustrates that sensors obtain  corrupted measurements of the temperature state, and communicate with other sensors   over a network to achieve an  estimate of the overall state.
		 }
	\label{fig:random_illustration}
\end{figure}

{ In a spatially distributed physical system, let the  state vector  consist of elements over a large geographical area.
	The evolution of the state is related to  spatial and temporal  dynamics.
	Sensors located at different positions can collaborate   based on their  intermittent measurements of partial elements of the state.
	The state   and the measurements  are   polluted by noise. 
	A  random  dynamic field   driven by noise $w_{k}$ and monitored by a sensor network is shown in Fig. \ref{fig:random_illustration}, cf. \cite{das2017consensus}. The variable $x_{k}^i$ stands for the temperature in station $i$ at time $k$. Colors   represent  values of $x_{k}^i$. 
	The problem considered in this paper is how to
	design a distributed robust filter based on  the corrupted measurements $y_{k,i},k\in\mathbb{N},i=1,\dots,4,$ and the collaboration of the sensors,  such that  the overall temperature field $x_{k}$ can be effectively estimated by each sensor.}

\subsection{Preliminaries}
Consider the system dynamics
\begin{equation}\label{system_all}
    x_{k+1}=(A_{k}+F_{k}\epsilon_k)x_{k}+w_{k},
\end{equation}
where $x_{k}\in \mathbb{R}^n$ denotes the system state vector, $w_{k}\in \mathbb{R}^n$   the independent process noise with zero mean, $\epsilon_k\in \mathbb{R}$   the independent multiplicative noise also with zero mean.  The matrices $F_{k}$, $k\in\mathbb{N}$, are non-singular matrices.

The system state is monitored by a sensor network with $N$ sensors
\begin{equation}\label{system_all2}
y_{k,i}=\gamma_{k,i}C_{k,i}x_{k}+v_{k,i}, i=1,\dots,N,
\end{equation}
where
 $y_{k,i}\in \mathbb{R}^{m_{i}}$ stands for  the measurement vector of sensor $i$, $v_{k,i}\in \mathbb{R}^{m_{i}}$   the independent measurement noise with zero mean and $\gamma_{k,i}\in \mathbb{R}$    the independent random fading factor   in the interval $[0,1]$ with $E\{\gamma_{k,i}\}=\tau_{k,i}$, where $0< \tau_{k,i}\leq 1$ is a known scalar, all at time $k=1,2,\dots.$.
%
The matrices $A_{k}$, $F_{k},$ and $C_{k,i}$ have   appropriate dimensions and are known to sensor $i$.

We model the sensor communications as a directed graph $\mathcal{G}=(\mathcal{V},\mathcal{E},\mathcal{A})$, which consists of   nodes
$\mathcal{V}=\{1,2,\ldots,N\}$,  links $\mathcal{E}\subseteq \mathcal{V}\times \mathcal{V}$, and the weighted adjacency matrix $\mathcal{A}=[a_{i,j}]$,
where $a_{i,i}> 0,a_{i,j}\geq 0,\sum_{j\in \mathcal{V}}a_{i,j}=1$. If $a_{i,j}>0,j\neq i$, there is a link $(j,i)\in \mathcal{E}$, through which node $i$ can directly receive  messages from node $j$. { In this case, node $j$ is called a (in-)neighbor of node $i$ and node $i$ is called a out-neighbor of node $j$. The (in-)neighbor set of node $i$, including itself,  is  denoted by $\mathcal{N}_{i}$.
	The graph
	$\mathcal{G}$ is called strongly connected if for any two nodes $i_{1},i_{l}$, there exists a directed path from $i_{l}$ to  $i_{1}$ : $(i_{l},i_{l-1}),\ldots,(i_{3},i_{2}),(i_{2},i_{1})$. }
Let $\{\tilde x_{k,j}$, $\tilde P_{k,j}\}$ be the pair that node $j$ communicates to its out-neighbor nodes  at time $k$, where $\tilde x_{k,j}\in\mathbb{R}^n$ and $\tilde P_{k,j}\in\mathbb{R}^{n\times n}$.
Due to  channel noise, the pair $\{\hat{\tilde x}_{k,i,j},\bar{\tilde P}_{k,i,j}\}$ received by node $i$ from node  $j$  is
\begin{align}\label{eq_corrupt1}
\begin{split}
\hat{\tilde x}_{k,i,j}&=\tilde x_{k,j}+\varepsilon_{k,i,j},j\in\mathcal{N}_{i}\\
\bar{\tilde P}_{k,i,j}&=\tilde P_{k,j}+\mathcal{D}_{k,i,j},j\in\mathcal{N}_{i},
\end{split}
\end{align}
where $\varepsilon_{k,i,j}\in \mathbb{R}^n$  and $\mathcal{D}_{k,i,j}\in\mathbb{R}^{n\times n} $ are the channel noise processes. { If $\tilde P_{k,j}$ is symmetric, $\mathcal{D}_{k,i,j}$ is reasonably assumed to be symmetric. Because  it is sufficient to transmit the  upper triangular part  of the symmetric matrix $\tilde P_{k,j}$. In Lemma \ref{lem_consistent}, we will show that  the transmitted matrix $\tilde P_{k,j}$ is indeed symmetric.}

Let $(\Omega,\mathcal{F},P)$ be  the basic probability space, and $\mathcal{F}_k$ be a filtration of the $\sigma$-algebra $\mathcal{F}$.  A discrete-time  sequence $\{\xi_k\}$ is said to be adapted to the family of $\sigma$-algebras $\{\mathcal{F}_k\}$ if $\xi_k$ is measurable to $\mathcal{F}_k$. We refer the reader 
 \cite{chow2012probability} for details.
We require the following assumption.
\begin{assumption}\label{ass_noise} 
	The following conditions  on noise and initial estimates hold.
\begin{enumerate} 
	\item The initial state $x_0$, its estimates $\hat x_{0,i}$, and the noise  $\epsilon_k$, $w_{k}$, $\gamma_{k+1,i}$,  $v_{k+1,i}$, $\varepsilon_{k+1,i,j}$, $\mathcal{D}_{k+1,i,j}$ are independent both in time and space, for all $i,j\in\mathcal{V}, k=0,1,\dots$. 
	\item There exist known matrices $Q_{k},R_{k+1,i},P_0$ and scalars $\mu_{k},\varphi_{k+1,i}$,  such that for all $i\in \mathcal{V}$, and $k=0,1,\dots,$
	\begin{align*}
	&E\{w_{k}w_{k}^T\}\leq Q_{k},\quad	\inf_{k\in\mathbb{N}} Q_{k}>0,\quad E\{x_0x_0^T\}\leq P_0\\[0.3em]
	&E\{\epsilon_k^2\}\leq \mu_{k},\quad \quad\quad\Cov\{\gamma_{k+1,i}\}\leq \varphi_{k+1,i}\\[0.3em]
	&E\{v_{k+1,i}v_{k+1,i}^T\}\leq R_{k+1,i}\\[0.3em]
	& \sup_{k\in\mathbb{N}} \left[\tau_{k+1,i}^2 C_{k+1,i}^TR_{k+1,i}^{-1}C_{k+1,i}\right]<\infty\\[0.3em]
	&E\{(\hat x_{0,i}-x_0)(\hat x_{0,i}-x_0)^T\}\leq P_{0,i}.
	\end{align*}
		\item  	There exist positive semi-definite matrices $\varUpsilon_{i,j}$ and $\mathcal{D}_{i,j}$ such that  for all $i\in \mathcal{V}$, $j\in\mathcal{N}_i,$ and $k=1,2,\dots,$ 
		\begin{align*}
&\sup\{\varepsilon_{k,i,j}\varepsilon_{k,i,j}^T\}\leq \varUpsilon_{i,j},-\mathcal{D}_{i,j}\leq\mathcal{D}_{k,i,j}\leq \mathcal{D}_{i,j},
		\end{align*}
		where the channel noise   $\varepsilon_{k,i,j}$ and $\mathcal{D}_{k,i,j}$ are in \eqref{eq_corrupt1}.
\end{enumerate}
\end{assumption}

Note that the exact covariance information of the stochastic uncertainties is not required. Bounds and statistics are  known only to individual sensors. Thus, the conditions in 2)  of Assumption \ref{ass_noise} are milder than \cite{liu2016minimum,li2017distributed,wen2016recursive}, where each sensor was assumed to have  full knowledge on the statistics of the  system.

Let $\hat x_{k}$ be the estimate of the  system state $x_{k}$.
 Due to  unknown correlation between sensor estimates,  the MSE of each sensor can not be obtained in a distributed manner \cite{Battistelli2014Kullback, Wang2017On,Yang2014Stochastic,yang2017stochastic}.  We introduce the following definitions to consider the bounds of MSE.
\begin{definition} \cite{Niehsen2002Information} (Consistency)\label{def_consistency}
The pair $\{\hat x_{k},P_{k}\}$ is consistent  if there is a deterministic sequence $\{P_{k}\}$ such that $	E\{(\hat x_{k}-x_{k})(\hat x_{k}-x_{k})^T\}\leq P_{k}.$	
\end{definition} 
\begin{definition}(Conditional consistency)\label{def_cond_consistency}
The pair $\{\hat x_{k},P_{k}\}$ is conditionally consistent  if there is a sequence $\{P_{k}\}$, such that 
$	E\{(\hat x_{k}-x_{k})(\hat x_{k}-x_{k})^T|\mathcal{K}_{k}\}\leq P_{k}$, where $\mathcal{K}_{k}$ is a $\sigma$-algebra and $P_{k}$ is measurable to $\mathcal{K}_{k}$.
\end{definition}
{ Note that the consistency defined  above is different from the one in parameter identification, which instead is on  asymptotic convergence  to the true parameters.}
{ 	The consistency definition we use in this paper \cite{he2019distributed,He2017Consistent,Battistelli2014Kullback,Wang2017On} provides two benefits. First, the estimation error of each sensor can be evaluated online by utilizing some probability inequalities \cite{poznyak2009advanced}. Second, a CI-based fusion method can be utilized in the filter design. We introduce conditional consistency in Definition~\ref{def_cond_consistency} to cope with channel noise.
	The idea is to use that   the pair $\{\hat x_{k},E\{P_{k}\}\}$ is  consistent, if $\{\hat x_{k},P_{k}\}$ is conditionally consistent.  }

\subsection{Problem}
In this paper, we consider a three-step distributed filtering structure. Each  sensor $ i\in \mathcal{V}$, executes a  state prediction, measurement update and local fusion at each time:
\begin{equation}\label{filter_stru}
\begin{split}
\bar x_{k,i}&=A_{k-1}\hat x_{k-1,i}\\[0.3em]
\tilde x_{k,i}&=\bar x_{k,i}+K_{k,i}(y_{k,i}-\tau_{k,i}C_{k,i}\bar x_{k,i})\\[0.3em]
\hat x_{k,i}&=\sum_{j\in \mathcal{N}_{i}}W_{k,i,j}\hat{\tilde x}_{k,i,j},
\end{split}
\end{equation}
where $\bar x_{k,i}$,  $\tilde x_{k,i}$, and $\hat x_{k,i}$ are the state estimates in  prediction,   update, and  fusion of sensor $i$ at time $k$, respectively. Moreover, $\hat{\tilde x}_{k,i,j}$ given in \eqref{eq_corrupt1} is the noisy estimate received by sensor $i$ from sensor $j$.
Besides, $K_{k,i}$ is the filtering gain parameter matrix, $W_{k,i,j}$ is the local fusion parameter  matrix. Both $K_{k,i}$ and $W_{k,i,j}$  remain to be designed.

	Different from the existing results \cite{Cattivelli2008Diffusion,Wang2017On,Khan2014Collaborative,liu2016minimum}, measurements and measurement matrices are not transmitted in our setting. The advantages of this protocol lie in several aspects including privacy, security   and energy saving.

In  this paper, we  consider three essential subproblems:
{ 
	
	(\textbf{a}) How to design the parameters $K_{k,i}$ and $W_{k,i,j}$ in the distributed filter (\ref{filter_stru}),  such that  the filter is conditionally consistent? (Lemmas \ref{lem_consistent} and \ref{lem_K})

	(\textbf{b}) Which conditions on system structure  and noise statistics    enable the mean square estimation error   to be  bounded? (Theorem \ref{theorem_boundedness})
	
	(\textbf{c}) How to improve the performance of the filter (\ref{filter_stru}) when past estimates are available? (Algorithm \ref{alg:B}, Proposition \ref{lem_com}, and Theorem~~\ref{theorem_boundedness2})
	}

\section{Distributed Robust Kalman Filter Design}\label{sec:filter}
In this section,  we first provide a distributed design  of the filter gain $K_{k,i}\text{ and the fusion weight }W_{k,i,j}$ of the filter  (\ref{filter_stru}). 
Then we present our proposed distributed robust Kalman filter (DRKF) algorithm. Finally, it is shown that the algorithm gives bounded MSE.

\begin{lemma}\label{prop_1}
	Under Assumption \ref{ass_noise}, it holds that $E\{x_{k}x_{k}^T\}\leq \varPi_{k}, \forall k\in\mathbb{N}$,
where $\varPi_{k}$ is recursively calculated through $\varPi_{k+1}=A_{k}\varPi_{k}A_{k}^T+\mu_{k} F_{k}\varPi_{k}F_{k}^T+ Q_{k},$
with $\varPi_{0}=P_0$, in which $P_0$, $\mu_{k}$, and $Q_{k}$  are in Assumption \ref{ass_noise}.
\end{lemma}
\begin{proof}
	See  Appendix \ref{app_lem_moment}.
\end{proof}

{Lemma \ref{prop_1} provides an upper bound of  the mean square of the system state $x(t)$, which is accessible to each sensor based on its system knowledge and 
	useful in the algorithm design and analysis as follows. Similar approaches are found in \cite{Tugnait1981Stability,yang2002robust}.}
By employing the CI-method \cite{Niehsen2002Information}, the following lemma provides a   choice for the fusion weight $W_{k,i,j}$ giving conditional consistency.
\begin{lemma}\label{lem_consistent}  
	Consider  system \eqref{system_all}--\eqref{system_all2} satisfying Assumption~\ref{ass_noise}. For the filter  (\ref{filter_stru}) with $k\geq 1$ and $i\in\mathcal{V}$, 
	if $K_{k,i}$ is adapted to the channel noise $\sigma$-algebra $\mathcal{W}_{k}=\sigma(\mathcal{D}_{t,i,j},1\leq t\leq k,i,j,\in\mathcal{V})$, and
		\begin{align}\label{eq_W}
	W_{k,i,j}=a_{i,j}P_{k,i}(\bar{\tilde P}_{k,i,j}+\mathcal{D}_{i,j}+\varUpsilon_{i,j})^{-1},
	\end{align}
then the pairs $\{\bar x_{k,i},\bar P_{k,i}\}, \{\tilde x_{k,i}, \tilde P_{k,i}\}, \{\hat x_{k,i},P_{k,i}\}$ are all conditionally consistent given $\mathcal{W}_{k}$, 
 where 
\begin{align*}
\bar P_{k,i}=&A_{k-1}P_{k-1,i}A_{k-1}^T+\mu_{k-1} F_{k-1}\varPi_{k-1}F_{k-1}^T+ Q_{k-1}\\[0.3em]
\tilde P_{k,i}=& (I-\tau_{k,i}K_{k,i}C_{k,i})\bar P_{k,i}(I-\tau_{k,i}K_{k,i}C_{k,i})^T\nonumber\\[0.3em]
&+K_{k,i}\big(R_{k,i}+\varphi_{k,i} C_{k,i}\varPi_{k}C_{k,i}^T\big)K_{k,i}^T\\[0.3em]
\bar{\tilde P}_{k,i,j}=&\tilde P_{k,j}+\mathcal{D}_{k,i,j},j\in\mathcal{N}_{i}\\[0.3em]
	P_{k,i}=&(\sum_{j\in \mathcal{N}_{i}}a_{i,j}(\bar{\tilde P}_{k,i,j}+\mathcal{D}_{i,j}+\varUpsilon_{i,j})^{-1})^{-1}.
\end{align*}
\end{lemma}
\begin{proof}
	See Appendix \ref{app_lem_consis}.
\end{proof}
{ Note that the design of the fusion weight $W_{k,i,j}$ in Lemma~\ref{lem_consistent} is fully distributed, and it	
	depends on the communication noise bounds, i.e., $\varUpsilon_{i,j}$ and $\mathcal{D}_{k,i,j}$, which is an extension to \cite{ji2017distributed,liu2016minimum,li2017distributed,wen2016recursive,He2017Consistent}.} In  the following lemma, we  design the filter gain $K_{k,i}$ of filter \eqref{filter_stru} such that the bound of the conditional MSE, i.e., $\tilde P_{k,i}$, is minimized at each measurement update.
\begin{lemma}\label{lem_K}
	The optimal solution $K_{k,i}^*:=\arg\min\limits_{K_{k,i}}\Tr\{\tilde P_{k,i}\}$ is given by 
\begin{align*}
K_{k,i}^*=\tau_{k,i}\bar P_{k,i}C_{k,i}^T\Xi_{k,i}^{-1},
\end{align*}
where $\Xi_{k,i}=\tau_{k,i}^2C_{k,i}\bar P_{k,i}C_{k,i}^T+R_{k,i}+\varphi_{k,i}C_{k,i}\varPi_{k}C_{k,i}^T.$ Furthermore, $K_{k,i}^*$ is adapted to the channel noise $\sigma$-algebra $\mathcal{W}_{k}$ in Lemma \ref{lem_consistent}.
\end{lemma}
\begin{proof}
		See Appendix \ref{app_lem_K}.
\end{proof}

{The designed filter gain in Lemma \ref{lem_K} inherits the gain of the optimal centralized robust filters in \cite{Tugnait1981Stability,yang2002robust}, but here it is stochastic and adapted to  the channel noise $\sigma$-algebra $\mathcal{W}_{k}=\sigma(\mathcal{D}_{t,i,j},1\leq t\leq k,i,j,\in\mathcal{V})$.}
With the filter  parameters $K_{k,i}$ and $W_{k,i,j}$ given in Lemmas \ref{lem_consistent} and \ref{lem_K}, respectively,  we obtain the DRKF given in Algorithm  \ref{alg:A}.   Different from \cite{Yang2014Stochastic,yang2017stochastic}, the implementation of this algorithm only depends on the local measurement information 
$\{y_{k,i},C_{k,i},R_{k,i},\varphi_{k,i},\tau_{k,i}\}$
and the estimate pairs $\{\hat{\tilde x}_{k,i,j}$, $\bar{\tilde P}_{k,i,j},j\in\mathcal{N}_i\}$ from neighbors. Thus, it obeys a fully distributed design and implementation. { 
	For  sensor $i$, the computational complexity of Algorithm  \ref{alg:A} at each time is $O(\max\{n^3d_i,m_i^3\})$, where $d_i$ is the cardinality of the set $\mathcal{N}_i$, and $n$ and $m_i$ are the dimensions of the system state and sensor measurement, respectively. The overall computational complexity  for all sensors is consequently $O(\max\{Nn^3d_i,Nm_i^3\})$. Thus, the algorithm is scalable to large networks. The performance of the algorithm is degraded if the upper bounds in Assumption \ref{ass_noise} are not  tight. 
}{  In systems with measurement outliers \cite{stojanovic2016joint}, 
Algorithm~\ref{alg:A} can be adapted to estimate the state  by developing appropriate scheme for discarding the measurement outliers. }

\begin{algorithm}
\caption{Distributed robust Kalman filter (DRKF):}
\label{alg:A}
\begin{algorithmic}
	{ 	\STATE {\textbf{Initial setting: }} \\[0.3em]
		$\{\hat x_{0,i},P_{0,i},\varPi_{0},\mathcal{D}_{i,j},\varUpsilon_{i,j},j\in\mathcal{N}_i,i\in\mathcal{V}\}$.}\\[0.3em]
\STATE {\textbf{Prediction:} For each sensor $i$:}\\[0.3em]
$\bar x_{k,i}=A_{k-1}\hat x_{k-1,i},$\\ [0.3em] 
		$\bar P_{k,i}=A_{k-1}P_{k-1,i}A_{k-1}^T+\mu_{k-1} F_{k-1}\varPi_{k-1}F_{k-1}^T+ Q_{k-1},$\\[0.3em]
		$	\varPi_{k}=A_{k-1}\varPi_{k-1}A_{k-1}^T+\mu_{k-1} F_{k-1}\varPi_{k-1}F_{k-1}^T+ Q_{k-1}.$\\[0.3em]
\STATE {\textbf{Update:} For each sensor $i$:}\\[0.3em]
$\tilde x_{k,i}=\bar x_{k,i}+K_{k,i}(y_{k,i}-\tau_{k,i}C_{k,i}\bar x_{k,i})$,\\        
\begin{flushleft}
$K_{k}=\tau_{k,i}\bar P_{k,i}C_{k,i}^T\big(\tau_{k,i}^2C_{k,i}\bar P_{k,i}C_{k,i}^T+R_{k,i}+\varphi_{k,i}C_{k,i}\varPi_{k}C_{k,i}^T\big)^{-1}$
\end{flushleft}
		$\tilde P_{k,i}=(I-\tau_{k,i}K_{k,i}C_{k,i})\bar P_{k,i}$.\\[0.3em]
\STATE {\textbf{Fusion:} For each sensor $i$:}\\[0.3em]
$\hat x_{k,i}=P_{k,i}\sum_{j\in \mathcal{N}_{i}}a_{i,j}(\bar{\tilde P}_{k,i,j}+\mathcal{D}_{i,j}+\varUpsilon_{i,j})^{-1}\hat{\tilde x}_{k,i,j}$,\\	[0.3em]
		$P_{k,i}=(\sum_{j\in \mathcal{N}_{i}}a_{i,j}(\bar{\tilde P}_{k,i,j}+\mathcal{D}_{i,j}+\varUpsilon_{i,j})^{-1})^{-1},$\\	[0.3em]
		where $\hat{\tilde x}_{k,i,j}$ and $\bar{\tilde P}_{k,i,j}$ are given in \eqref{eq_corrupt1}.
\end{algorithmic}
\end{algorithm}

Next  we   find mild conditions to guarantee  boundedness of the MSE for Algorithm \ref{alg:A}. For $j> k$, we denote  the  transition matrix by $\Phi_{j,k}= A_{j-1}\Phi_{j-1,k}$, where $\Phi_{k,k}=I_{n}$. We assume robust collective observability in the following.

\begin{assumption}\label{ass_observable}(Robust collective  observability)
	There exists an integer $\bar N>0$ and a constant $\alpha>0$ such that for $k\in\mathbb{N}$,
\begin{equation}\label{Observability_matrix2}
\sum_{i=1}^{N}\sum_{j=k}^{k+\bar N}\Phi^T_{j,k}\bar C_{j,i}^T\tilde R_{j,i}^{-1}\bar C_{j,i}\Phi_{j,k}\geq \alpha I_n,
\end{equation}	
where  
\begin{align*}
\bar C_{j,i}&= \tau_{j,i}C_{j,i},\quad  j\in\mathbb{N}, \quad i\in\mathcal{V}\\[0.3em]
\tilde R_{j,i}&= R_{j,i}+\varpi_{j}\varphi_{j,i}C_{j,i}C_{j,i}^T\\[0.3em]
\varpi_{j}&=\|P_{0}\|_2\prod_{i=0}^{j-1}\bar\alpha_{i} + \sum_{s=1}^{j}\left(\bar q_{s-1}\prod_{l=s}^{j}\bar\alpha_{l} \right)+ \bar q_{j}\\[0.3em]
\bar \alpha_{j}&=\|A_{j}\|_2^2+\mu_{j}\|F_{j}\|_2^2\\[0.3em]
\bar q_{j}&=\|Q_{j}\|_2.
\end{align*}
\end{assumption}
Assumption \ref{ass_observable} is based on the system structure and noise statistics. It can be regarded as a distributed version of  the observability condition with multiplicative noise in \cite{Tugnait1981Stability}.
The condition does not require that each sub-system is observable \cite{Cat2010Diffusion,stankovic2009consensus,Yang2014Stochastic}.
Moreover, if $\varphi_{k,i}\equiv 0$, $\forall k\in\mathbb{N},i\in\mathcal{V}$,  Assumption \ref{ass_observable} corresponds to the collective observability condition for time-varying stochastic systems in \cite{He2017Consistent}.

	A requirement on the multiplicative noise $\epsilon_k$ is needed.  Recall that $\mu_k$ is the bound of the variance of $\epsilon_k$.
 Denote
 the time sequence  of  non-zero multiplicative noise by
 \begin{align}\label{noise_set}
 \mathbb{K}_{T}=\{k_{t}=\min\limits_{ \mu_{k}>0} k|k\geq k_{t-1},k,t\in\mathbb{N}\}.
 \end{align}
 \begin{assumption}\label{ass_structure} 
There exist positive scalars $\lambda_1$, $\lambda_2$, $M$ and $\varrho\in(0,1)$, such that 
 \begin{align}
 &\lambda_1 I_n\leq A_{k}A_{k}^T\leq \lambda_2 I_n, k\in\mathbb{N}\label{eq_A}\\[0.3em]
 &\prod_{t=s}^l\rho_{k_t}\leq M\varrho^{l-s}, 0\leq s\leq l< \infty\label{eq_exp}\\[0.3em]
 &\sup\limits_{t\in\mathbb{N}}\|\mu_{k_{t+1}}F_{k_{t+1}}\mathcal Q_{k_{t+1},k_{t}}F_{k_{t+1}}^T\|_2<\infty,\label{eq_Q_sup}
 \end{align}	
 where $k_t\in \mathbb{K}_{T}$ in \eqref{noise_set} and
 \begin{align*}
 &\rho_{k_t}=\frac{\mu_{k_{t+1}}}{\mu_{k_{t}}}\|F_{k_{t+1}}\Phi_{k_{t+1},k_{t}}F_{k_{t}}^{-1}\|_{2}^2+\mu_{k_{t+1}}\|F_{k_{t+1}}\Phi_{k_{t+1},k_{t}}\|_2^2\\[0.3em]
 &\mathcal Q_{k_{t+1},k_{t}}=\sum_{k=k_{t}}^{k_{t+1}}\Phi_{k_{t+1},k}Q_{k}\Phi_{k_{t+1},k}^T.
 \end{align*}
 \end{assumption}
		Compared to \cite{liu2016minimum,li2017distributed,wen2016recursive}, \eqref{eq_A} is a  milder condition as it permits the nominal system to be unstable. 	If   $\{k|\mu_{k}>0,k\in\mathbb{N}\}$ is finite or even empty, 
		\eqref{eq_exp} and \eqref{eq_Q_sup} can still be made satisfied by replacing the points $\mu_{k}=0$ with  sufficiently small positive $\bar\mu_{k}$. 	
For further analysis, we need  Lemmas \ref{lem_sufficient_obser}--\ref{lem_multi_bounds}.
\begin{lemma}\label{lem_sufficient_obser}
	If Assumption \ref{ass_observable} holds, then 
		\begin{equation}\label{Observability_matrix}
			\sum_{i=1}^{N}\sum_{j=k}^{k+\bar N}\Phi^T_{j,k}\bar C_{j,i}^T\bar R_{j,i}^{-1}\bar C_{j,i}\Phi_{j,k}\geq \alpha I_n,
			\end{equation}	
			where $ \bar R_{k,i}:= R_{k,i}+\varphi_{k,i}C_{k,i}\varPi_{k}C_{k,i}^T$.
\end{lemma}
\begin{proof}
See Appendix \ref{app_lem_obser}.
\end{proof}
{Different from \eqref{Observability_matrix2} in Assumption \ref{ass_observable},  \eqref{Observability_matrix} in Lemma \ref{lem_sufficient_obser} utilizes $\varPi_{k}$ given by Lemma \ref{prop_1}.
	We note that Lemma \ref{lem_sufficient_obser} is provided  for the proof of Theorem \ref{theorem_boundedness}.}
\begin{lemma}\label{lem_multi_bounds}
		If Assumption \ref{ass_structure} holds, then 
			 	$$\sup\limits_{k\in\mathbb{N}}\{\mu_{k} F_{k}\varPi_{k}F_{k}^T\}<\infty.$$
\end{lemma}
\begin{proof}
See Appendix \ref{app_lem_sufficient}.
\end{proof}
{Lemma \ref{lem_multi_bounds} is useful in the proof of the following theorem.}
Next we state our main result on Algorithm \ref{alg:A}: the estimation MSE of $e_{k,i}:=\hat{x}_{k,i}-x_{k} $ is bounded.
\begin{theorem}\label{theorem_boundedness} 
{ 	Suppose  system \eqref{system_all}--\eqref{system_all2} satisfies Assumptions~\ref{ass_noise}, \ref{ass_observable}--\ref{ass_structure} and that $\mathcal{G}$ is strongly connected. Then, the estimation MSE for Algorithm \ref{alg:A} is uniformly  bounded for all sensors, i.e., 
	there exists a positive scalar $ \eta$ such that}
\begin{align*}
\sup_{k\geq   N+\bar N}\lambda_{max}\left( E\{e_{k,i}e_{k,i}^T\}\right) \leq \frac{\eta}{\alpha}, \forall i\in\mathcal{V},
\end{align*}
	where $\alpha$ is given in Assumption \ref{ass_observable}.	
\end{theorem}
\begin{proof}
See Appendix \ref{app_thm_bound}.
\end{proof}

Theorem \ref{theorem_boundedness} states that a larger  $\alpha$ can lead to a smaller upper bound of the  MSE. Thus, increasing observability ($\bar C_{k,i}$) and reducing noise interference ($\bar R_{k,i}$) can both contribute to improving  estimation performance of the DRKF in Algorithm~~\ref{alg:A}.

\section{DRKF with  a sliding window}
{ In this section, we modify the DRKF algorithm to include also past estimates received from neighbors. The presented DRKF with sliding-window fusion (DRKF-SWF) algorithm is shown to give bounded MSE. In the numerical simulation in next section, it is shown to sometimes outperform the DRKF algorithm.}

Since  the  estimates $\{\hat{\tilde x}_{k,i,j}$, $\bar{\tilde P}_{k,i,j},j\in\mathcal{N}_i\}$  have been corrupted by  the channel noise through \eqref{eq_corrupt1}, designing a distributed filter simply based on the latest estimates may lead to performance degradation if these estimates have been seriously deteriorated. In this case, we  fuse the past estimates received  from neighbors. This leads to  a better estimate than that of simply fusing current estimates. 
To decide which past estimates to use, a sliding window with length $L\geq 1$ is introduced.
For $l=0,\dots,L,$ we denote 
\begin{equation}\label{eq_notation}
\begin{split}
&\check x_{k-l,j}:=\hat{\tilde x}_{k-l,i,j}\\[0.3em]
&\check P_{k-l,j}:=\bar{\tilde P}_{k-l,i,j}+\mathcal{D}_{i,j}+\varUpsilon_{i,j}.
\end{split}
\end{equation}
By Lemma \ref{lem_consistent}, $\{\check x_{k,j},\check P_{k-l,j}\}$ is  conditionally consistent given the channel noise $\sigma$-algebra $\mathcal{W}_{k}=\sigma(\mathcal{D}_{t,i,j},1\leq t\leq k,i,j,\in\mathcal{V})$.
Sensor $i$  has the available messages $\{\check x_{l,j},\check P_{l,j}\}_{l=k-L+1}^{k}$ from sensor $j$. 
We denote
\begin{align}\label{eq_pred}
\begin{split}
&(\check x_{k,j}^{1},\check P_{k,j}^{1}):=(f_{0}(\check x_{k,j}),g_{0}(\check P_{k,j})):=(\check x_{k,j},\check P_{k,j})\\[0.3em]
&(\check x_{k,j}^{2},\check P_{k,j}^{2}):=(f_{1}(\check x_{k-1,j}),g_{1}(\check P_{k-1,j}))\\[0.3em]
&\qquad\qquad\quad\vdots\\[0.3em]
&(\check x_{k,j}^{L},\check P_{k,j}^{L}):=(f_{L-1}(\check x_{k-L+1,j}),g_{L-1}(\check P_{k-L+1,j})),
\end{split}
\end{align}
where for $l=1,\dots,L-1$,
\begin{align}\label{eq_pre_operator}
\begin{split}
f_{l}(\check x_{k-l,j})&=f_{1}(f_{l-1}(\check x_{k-l,j}))\\[0.3em]
g_{l}(\check P_{k-l,j})&=g_{1}(g_{l-1}(\check P_{k-l,j}))\\[0.3em]
f_{1}(\check x_{k-l,j})&=A_{k-l}\check x_{k-l,j}\\[0.3em]
g_{1}(\check P_{k-l,j})&=A_{k-l}\check P_{k-l,j}A_{k-l}^T+ Q_{k-l}\\[0.3em]
&\quad+\mu_{k-l} F_{k-l}\varPi_{k-l}F_{k-l}^T.
\end{split}
\end{align}
At time $k$, based on the local knowledge and the information received from neighbors, sensor $i$  can fuse the messages $\{\check x_{l,j},\check P_{l,j},j\in\mathcal{V}_i\}_{l=k-L+1}^{k}$   to obtain a better estimate of $x_{k}$. By \eqref{error_notations}, $\{\check x_{l,j},\check P_{l,j},j\in\mathcal{V}_i\}_{l=k-L+1}^{k}$ are all conditionally consistent given  $\mathcal{W}_{k}=\sigma(\mathcal{D}_{t,i,j},1\leq t\leq k,i,j,\in\mathcal{V})$.

Let 
\begin{align}\label{eq_fused_2}
&\hat x_{k,i}=P_{k,i}\sum_{s=1}^L\sum_{j\in \mathcal{N}_{i}}a_{i,j,k}^{s}(\check P_{k,j}^{s})^{-1}\check x_{k,j}^{s}\\	
&P_{k,i}=\bigg(\sum_{s=1}^L\sum_{j\in \mathcal{N}_{i}}a_{i,j,k}^{s}(\check P_{k,j}^{s})^{-1}\bigg)^{-1},
\end{align}
where $a_{i,j,k}^{s}$ is  element $(i,j)$  of $\bar{\mathcal{A}_{k}}\in\mathbb{R}^{N\times NL}$ which is the CI fusion weight matrix for $\{\check x_{k,j}^{s},\check P_{k,j}^{s},j\in\mathcal{V}_i\}_{l=k-L+1}^{k}$. In the following, the design of $\bar{\mathcal{A}_{k}}$ is studied.
By the proof of Lemma~\ref{lem_consistent} and \eqref{eq_pred}, 	$\{\check x_{k,j}^{s},\check P_{k,j}^{s},j\in\mathcal{V}_i\}_{l=k-L+1}^{k}$ are  conditionally consistent given $\mathcal{W}_{k}=\sigma(\mathcal{D}_{t,i,j},1\leq t\leq k,i,j,\in\mathcal{V})$.
The design of $\bar{\mathcal{A}_{k}}$ is given by solving the following optimization problem.

{ 
\begin{align}\label{optim_11}
\begin{split}
	\underset{a_{i,j,k}^{s},j\in\mathcal{N}_i}{\text{minimize}}\qquad& \Tr(\mathcal{J}_{k,i}^{-1})\\
	\text{subject to}\qquad&\\
	& \mathcal{J}_{k,i}>0\\[0.3em]
	& 0\leq a_{i,j,k}^{s}\leq 1,\\[0.3em]
	& \sum_{s=1}^L\sum_{j\in \mathcal{N}_{i}}a_{i,j,k}^{s}=1
\end{split}
\end{align}
where $\mathcal{J}_{k,i}=\sum_{s=1}^L\sum_{j\in \mathcal{N}_{i}}a_{i,j,k}^{s}(\check P_{k,j}^{s})^{-1}-
\sum_{j\in \mathcal{N}_{i}}a_{i,j}\check P_{k,j}^{-1}.$
The optimal solution to \eqref{optim_11} is denoted by $\bar a_{i,j,k}^s,j\in \mathcal{N}_{i},s=1,\dots,L$.
According to \cite{He2017Consistent}, the problem in \eqref{optim_11}  is  convex and equivalent  to  an SDP problem, which   can be effectively solved by many existing algorithms if the problem is feasible. If the problem is infeasible, we use the same fusion 
approach as Algorithm \ref{alg:A}, i.e.,  $ \bar{\mathcal{A}_{k}}=\left(\begin{matrix}
\mathcal{A}&0^{N\times (N-1)L}
\end{matrix}\right)$. The
 feasibility of the SDP is equivalent to the feasibility test problem of linear matrix inequality \cite{boyd1994linear}.}
Due to resource constraints, it may be undesirable to solve the online optimization problem \eqref{optim_11} at each time.  Suppose  sensor $i$ has the ability to solve \eqref{optim_11} at time instants $\{k_s\}_{s=1}^{\infty}$, subject to 
\begin{align*}
\mod(k_s,\Delta_{i})=0,
\end{align*}
 where $\mod(a,b)$ is the remainder operator of $a/b$ and $\Delta_{i}\in\mathbb{Z}^+$ is the time interval length within which sensor $i$ can not solve the optimization problem. In other words, at time instants $\{k_s\}_{s=1}^{\infty}$, each sensor employs \eqref{eq_fused_2} to obtain a fused estimate, and for other instants, it utilizes the fusion methods in Algorithm \ref{alg:A} based on the latest estimates from its neighbors.
We provide the DRKF-SWF  in Algorithm \ref{alg:B}. {  Compared with \cite{ji2017distributed,liu2016minimum,li2017distributed,wen2016recursive,He2017Consistent,he2019distributed}, Algorithm \ref{alg:B}  utilizes the past information more efficiently and considers the limitation of step-wise optimization.
	 	  The computational burden of  Algorithm \ref{alg:B}, in addition to that of Algorithm \ref{alg:A}, is that it solves the SDP convex  optimization problem \eqref{optim_11} for every $\Delta_{i}$. Thus, also Algorithm~\ref{alg:B} scales to large networks, as such optimization problems are easy to solve.	 
	The difficulty in the implementation of Algorithm \ref{alg:B} is that solving the optimization problem \eqref{optim_11} needs more computational resources if the dimension of the system state increases.
}
\begin{algorithm}
	\caption{Distributed robust Kalman filter with sliding-window fusion  (DRKF-SWF):}
	\label{alg:B}
	\begin{algorithmic}
	{ 	\STATE {\textbf{Initial setting: } \\[0.3em]
			$\{L,\Delta_{i},\hat x_{0,i},P_{0,i},\varPi_{0},\mathcal{D}_{i,j},\varUpsilon_{i,j},j\in\mathcal{N}_i,i\in\mathcal{V}\}$.\\[0.3em]}}
		\STATE {\textbf{Prediction:} Same as Algorithm \ref{alg:A}.}\\[0.3em]
		\STATE {\textbf{Update:} Same as Algorithm \ref{alg:A}.}\\[0.3em]
		\STATE {\textbf{Local Fusion:} For each sensor $i$:} \\[0.3em]
		\textbf{if} $\mod(k,\Delta_{i})=0$ and \eqref{optim_11} has a feasible solution:  \\[0.3em]		
		$\qquad\hat x_{k,i}=P_{k,i}\sum_{s=1}^L\sum_{j\in \mathcal{N}_{i}}\bar a_{i,j,k}^{s}(\check P_{k,j}^{s})^{-1}\check x_{k,j}^{s}$\\[0.3em]	
		$\qquad P_{k,i}=\bigg(\sum_{s=1}^L\sum_{j\in \mathcal{N}_{i}}\bar a_{i,j,k}^{s}(\check P_{k,j}^{s})^{-1}\bigg)^{-1}$,\\
		\hangafter 1 where  $\check P_{k,j}^{s}$ and $\check x_{k,j}^{s}$ are given in (\ref{eq_pred}), and $\{\bar a_{i,j,k}^{s}\}_{s=1}^L$ are given by solving \eqref{optim_11};\\[0.3em]
		\textbf{else} \\[0.3em]
		$\qquad\hat x_{k,i}=P_{k,i}\sum_{j\in \mathcal{N}_{i}}a_{i,j}\big(\bar{\tilde P}_{k,i,j}+\mathcal{D}_{i,j}+\varUpsilon_{i,j}\big)^{-1}\hat{\tilde x}_{k,i,j}$\\	
		$\qquad P_{k,i}=\bigg(\sum_{j\in \mathcal{N}_{i}}a_{i,j}(\bar{\tilde P}_{k,i,j}+\mathcal{D}_{i,j}+\varUpsilon_{i,j})^{-1}\bigg)^{-1}$,\\[0.3em]
				where $\hat{\tilde x}_{k,i,j}$ and $\bar{\tilde P}_{k,i,j}$ are given in \eqref{eq_corrupt1}.
	\end{algorithmic}
\end{algorithm}

The following lemma shows that Algorithm \ref{alg:B} is  conditionally consistent given the channel noise $\sigma$-algebra $\mathcal{W}_{k}$.
\begin{lemma}\label{lem_consistent2}  
	Consider  system \eqref{system_all}--\eqref{system_all2} satisfying Assumption~\ref{ass_noise}. Then for Algorithm \ref{alg:B}, 
	 the pairs $\{\bar x_{k,i},\bar P_{k,i}\}$, $\{\tilde x_{k,i}, \tilde P_{k,i}\}$, and $\{\hat x_{k,i},P_{k,i}\}$ are   conditionally consistent given $\mathcal{W}_{k}$.
\end{lemma}
\begin{proof}
	Similar to the proof of Lemma \ref{lem_consistent} but considering the CI fusion in \eqref{eq_fused_2} and the fact that $K_{k,i}$ is adapted to $\mathcal{W}_{k}=\sigma(\mathcal{D}_{t,i,j},1\leq t\leq k,i,j,\in\mathcal{V})$. 
\end{proof}
{Lemma~\ref{lem_consistent2}, corresponding to   Lemma~\ref{lem_consistent}, shows that Algorithm~\ref{alg:B} shares the same conditional consistency as Algorithm~\ref{alg:A}. }
Algorithm~\ref{alg:B} is better than Algorithm~\ref{alg:A} in the following sense.
{ 
\begin{proposition}\label{lem_com}
		Consider  system \eqref{system_all}--\eqref{system_all2} satisfying Assumption \ref{ass_noise}. 
Under the same initial setting and the channel noise  $\sigma$-algebra $\mathcal{W}_{k}=\sigma(\mathcal{D}_{t,i,j},1\leq t\leq k,i,j,\in\mathcal{V})$, for Algorithms \ref{alg:A}--\ref{alg:B}, it holds that
\begin{align}\label{eq_ine}
P_{k,i}^{B}\leq P_{k,i}^{A},
\end{align}
where $P_{k,i}^{A}$ and $P_{k,i}^{B}$ are the $P_{k,i}$ matrix of Algorithm \ref{alg:A} and Algorithm \ref{alg:B}, respectively. 
\end{proposition} 
}
\begin{proof}
	If  $\mod(k,\Delta_{i})=0$ and \eqref{optim_11}  is feasible,  the constraint of \eqref{optim_11} $\mathcal{J}_{k,i}>0$  ensures that Algorithm \ref{alg:B} has a smaller $P_{k,i}$. Otherwise,   the fusion scheme of Algorithm \ref{alg:B} is the same as Algorithm \ref{alg:A}, which also ensures \eqref{eq_ine}.
\end{proof}

Proposition \ref{lem_com} shows that compared to Algorithm \ref{alg:A}, Algorithm \ref{alg:B}   has a smaller upper bound of the MSE.
	A larger window parameter $L$ can lead to a smaller objective function of \eqref{optim_11}, but the computation will increase as well.
Also,  the time length $\Delta_{i}$ influences the estimation performance, since a larger~$\Delta_{i}$ means that sensor~$i$ does not solve the optimization problem \eqref{optim_11} for a longer time interval.
	{ The parameters $L$ and $\Delta_{i}$ can be chosen based on the computational and communication ability of the sensor network. Furthermore, let $T$ be the time length of interest, then Algorithm \ref{alg:B} degenerates to Algorithm \ref{alg:A} if $\Delta_{i}>T$.  The boundedness of the MSE for Algorithm \ref{alg:B} is presented in the following.}

\begin{theorem}\label{theorem_boundedness2}
		Suppose  system \eqref{system_all}--\eqref{system_all2} satisfies Assumptions~\ref{ass_noise}, \ref{ass_observable}--\ref{ass_structure} and that $\mathcal{G}$ is strongly connected. Then, the estimation MSE for Algorithm \ref{alg:B} is uniformly  bounded for all sensors,  , i.e., there exists a positive scalar $\tilde \eta$ such that
\begin{align*}
\sup_{k\geq   N+\bar N}\lambda_{max}\left( E\{e_{k,i}e_{k,i}^T\}\right) \leq \frac{\tilde \eta}{\alpha}, \forall i\in\mathcal{V},
\end{align*}
where $\alpha$ is given in Assumption \ref{ass_observable}.	
\end{theorem}
\begin{proof}
It follows from Lemma \ref{lem_com} and the proof of Theorem~\ref{theorem_boundedness}.
\end{proof}
{Theorem \ref{theorem_boundedness2}, corresponding to  Theorem \ref{theorem_boundedness}, shows that Algorithm \ref{alg:B} shares the same MSE boundedness  as Algorithm~\ref{alg:A} under mild conditions. }

\section{Numerical Simulations}
In this section, we study two examples to  validate the effectiveness of the proposed algorithms and the theoretical results developed in the paper.
\subsection{Example 1}
For the temperature  field  in Fig. \ref{fig:random_illustration}, we suppose that the initial state $x_0$ and sensor measurement noise  are generated by  independent standard normal distributions. 
The fading factors  $\gamma_{k,i}$ follow independent uniform distributions, $i=1,2,3,4$. The time sequence $\{t_k\}$ lies in the interval $[0,10]$ with uniform sampling step $0.1$, thus $k=0,1,\dots,100.$ 
The matrices and scalars in  (\ref{system_all})  are assumed to be
\begin{align}\label{set1}
&{ A_{k}=\left(
	\begin{array}{cc}
	0.8\times(1+0.01 t_k) & 0.01\\
	0.1 & 0.98\\
	\end{array}
	\right)\nonumber}\\
&F_{k}=I_4,Q_{k}=0.1\times I_2,P_0=I_2,\mu_{k}=0.1\times (t_k+2)^{-1}\nonumber\\
&R_{k,1}=0.07,R_{k,2}=0.08,R_{k,3}=R_{k,4}=0.09\nonumber\\
&\tau_{k,1}=0.85,\varphi_{k,1}=0.8\times 10^{-3},C_{k,1}=\left(
\begin{array}{cc}
0 & 1\\
\end{array}
\right)\\
&\tau_{k,2}=0.15,\varphi_{k,2}=0.8\times 10^{-3},C_{k,2}=\left(
\begin{array}{cc}
0 & 1\\
\end{array}
\right)\nonumber\\
&\tau_{k,3}=0.20,\varphi_{k,3}=0.8\times 10^{-3},C_{k,3}=\left(
\begin{array}{cc}
0 & 1\\
\end{array}
\right)\nonumber\\
&\tau_{k,4}=0.85,\varphi_{k,4}=0.8\times 10^{-3},C_{k,4}=\left(
\begin{array}{cc}
1 & 0\\
\end{array}
\right).\nonumber
\end{align}

The initial  setting of the filters is $\hat x_{i,0}=\textbf{1}_2$ and $P_{i,0}=100\times I_2$, $\forall i\in\mathcal{V}$. The weighted adjacency matrix is
\begin{align*}
\mathcal{A}=[a_{i,j}]=\left(\begin{matrix}
0.3&0.7&0&0\\
0&0.4&0.6&0\\
0&0&0.3&0.7\\
0.3&0.4&0&0.3
\end{matrix}\right).
\end{align*}  
The channel noise is assumed to be mutually independent and uniformly distributed over $[-1,1]$. We choose $\varUpsilon_{i,j}=\mathcal{D}_{i,j}=I_2,i,j\in\mathcal{V}$.
\begin{figure}
	\centering
	\includegraphics[width=0.5\textwidth]{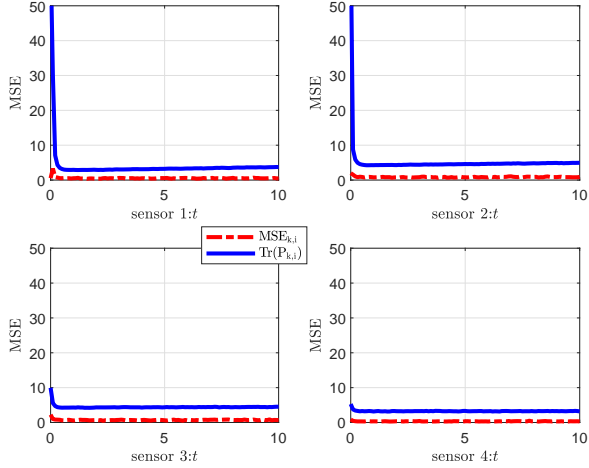}
	\caption{Consistent estimates of DRKF}
	\label{consistency_DRKF}
\end{figure}
\begin{figure}
	\centering
	\includegraphics[width=0.5\textwidth]{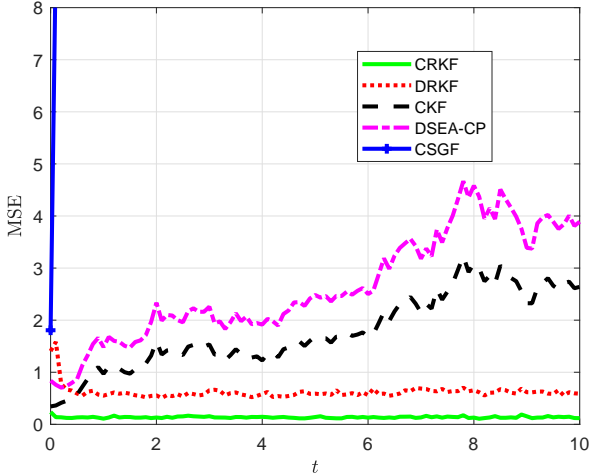}
	\caption{Comparison of tracking performance for the proposed filter DRKF together with filters from the literature}
	\label{performance_compare}
\end{figure}
\begin{figure}
	\centering
	\includegraphics[width=0.5\textwidth]{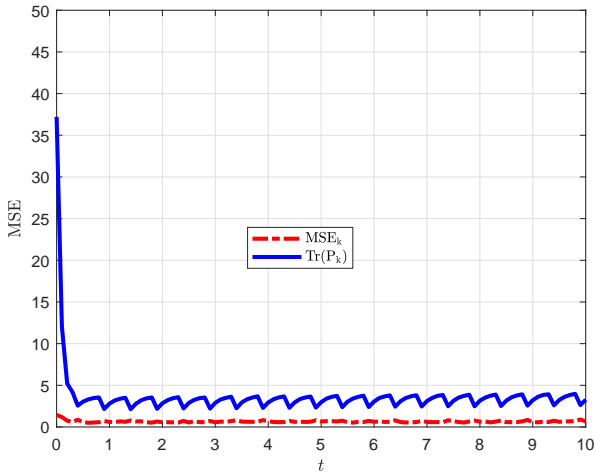}
	\caption{Consistent estimates of  DRKF-SWF with $L=2$ and $\Delta_{i}=5$}
	\label{consistency_DRKF2}
\end{figure}
\begin{figure}
	\centering
	\includegraphics[width=0.5\textwidth]{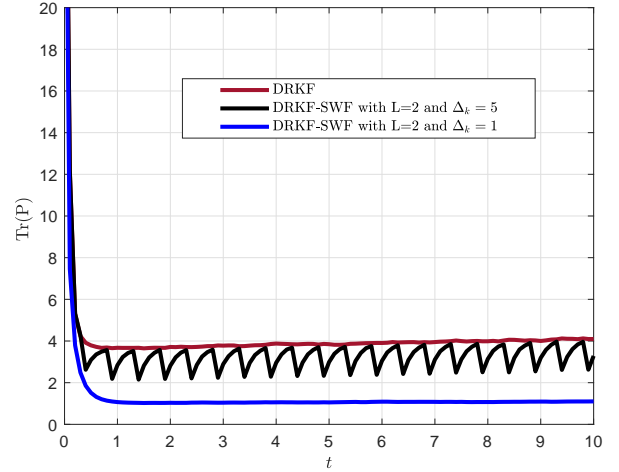}
	\caption{Comparison between DRKF and DRKF-SWF}
	\label{compare_A12}
\end{figure}
We conduct Monte Carlo experiments, in which $N_t=100$  runs  are performed.  
We denote
{ 
	\begin{align}\label{objective_sim}
	\begin{split}
	\rm{MSE}_{k,i}&=\frac{1}{N_t}\sum_{j=1}^{N_t}(\hat x_{k,i}^j-x_{k}^j)^T(\hat x_{k,i}^j-x_{k}^j)\\
	\Tr(P_{k,i})&=\frac{1}{N_t}\sum_{j=1}^{N_t}\Tr(P_{k,i}^j),
	\end{split}
	\end{align}}
where $\hat x_{k,i}^j$ and $P_{k,i}^j$ are the state estimate and parameter matrix of the $j$th run of sensor $i$. 

We show how $\Tr(P_{k,i})$ is an upper bound of $\rm{MSE}_{k,i}$. Fig. \ref{consistency_DRKF}   shows that this  holds for each sensor. Let $\rm{MSE}_{k}=\frac{1}{|\mathcal{V}|}\sum_{i\in\mathcal{V}}\rm{MSE}_{k,i}$, $\Tr(P_{k})=\frac{1}{|\mathcal{V}|}\sum_{i\in\mathcal{V}}\Tr(P_{k,i})$.
{ To illustrate the relationship between the initial conditions and the output of the DRKF, we provide  Table \ref{table_con}, where 
	$\rm{MSE}_{\max}=\max\limits_{k=51,\dots,100} \rm{MSE}_k$, and $P_{\max}=\max\limits_{k=51,\dots,100} P_k$. Here we just consider  $k\in \{51,\dots,100\}$, since the estimation error after $k=51$ is relatively steady.  Table \ref{table_con} shows that $P_{0,i}$ and $\varPi_{0}$ have little influence on the output of the DRKF, but $\mathcal{D}_{i,j}$ and $\mathcal{\varUpsilon}_{i,j}$   affect $ \rm{MSE}_{\max}$ and $P_{\max}$, as expected. } 
\begin{table}[htbp]
	{ 
	\caption{$\rm{MSE}_{\max}$ and $P_{\max}$ of the DRKF with different initial Quantities}  
	\label{table_con}   
	\scalebox{1}{
		{\renewcommand{\arraystretch}{1.0}  
			\begin{tabular}{|c|c|c|c|c|c|c|} 
				\hline   
				Case number  & $P_{0,i}$&$\varPi_{0}$ &$\mathcal{D}_{i,j}$ & $\mathcal{\varUpsilon}_{i,j}$&$ \rm{MSE}_{\max}$&$P_{\max}$\\
				\hline
				1 & $100I_2$& $I_2$ & $I_2$& $I_2$&0.74&4.15\\					\hline
				2 & $500I_2$& $I_2$ & $I_2$& $I_2$&0.75&4.15\\					\hline
				3 & $100I_2$& $5I_2$ & $I_2$& $I_2$&0.73&4.16\\					\hline
				4 & $100I_2$& $I_2$ & $5I_2$& $I_2$&0.89&9.38\\					\hline
				5 & $100I_2$& $I_2$ & $I_2$& $5I_2$&0.90&9.38\\
				\hline  
			\end{tabular}
		}
	}
}
\end{table}

We compare the proposed DRKF algorithm with   centralized Kalman filter (CKF), centralized robust Kalman filter (CRKF) \cite{Tugnait1981Stability,yang2002robust}, collaborative scalar-gain estimator (CSGF) \cite{Khan2014Collaborative}, and distributed state
estimation with consensus on the posteriors (DSEA-CP) \cite{Battistelli2014Kullback}.  {  The centralized filters CKF and CRKF  utilize the observations of all sensors without suffering communication noise.  Moreover, for the considered scenario, CRKF
is the 	optimal robust filter in the sense that its filter gain  ensures the minimum bound of MSE \cite{Tugnait1981Stability,yang2002robust}.}
The MSE of these  algorithms  are shown in Fig. \ref{performance_compare}, which indicates that the DRKF achieves better estimation accuracy than  CSGF, DSEA-CP, and DRKF.
Fig. \ref{consistency_DRKF2} shows that  DRKF-SWF provides  bounded mean square estimation errors and consistent estimates. 
By setting $\Delta_{i}=\Delta$, $i\in\mathcal{V}$, Fig. \ref{compare_A12} shows that   DRKF-SWF with sliding-window length $L=2$ provides  smaller upper bounds than the DRKF by decreasing the interval length $\Delta_i$.

{ \subsection{Example 2}
Consider the undirected  network with 50 sensors  in Fig. \ref{network2}. The weights of the adjacency matrix are given by
\begin{align*}
a_{i,j}=&\frac{1}{\max\{d_i,d_j\}},\quad i\in\mathcal{V},j\in\mathcal{N}_i,j\neq i\\[0.3em]
a_{i,i}=&1-\underset{j\in\mathcal{N}_i,j\neq i}{\sum}a_{i,j}.
\end{align*}
where $d_i$ is the cardinality of the set $\mathcal{N}_i$.
	We assume $A_k=\begin{pmatrix}
	1.05& -0.1\\
	0.1 & 0.98
	\end{pmatrix}$, $\mu_{k}=0$, $R_{k,i}=1,i\in\{1,\dots,50\}$.   For each sensor, the pair of measurement vector and fading statistics  are randomly chosen out of the four combinations in \eqref{set1}. The rest of the simulation settings  are the same as in Example~1. Fig.~\ref{consistency2} shows the bounded MSE and its upper bound, which verifies the estimation consistency of  Algorithm~\ref{alg:A}.  In Fig. \ref{compare2}, we   compare the estimation performance of the DRKF with the four algorithms mentioned in Example~1. The result shows that the proposed DRKF achieves better performance than the CSGF, DSEA-CP, and CKF, whose estimation errors are diverging fast due to the instability of the system dynamics (i.e., $\rho(A_{k})=1.02>1$). The performance of Algorithm \ref{alg:A}, i.e., DRKF, is close to CRKF.
}

\begin{figure}
	\centering
	\includegraphics[width=0.5\textwidth]{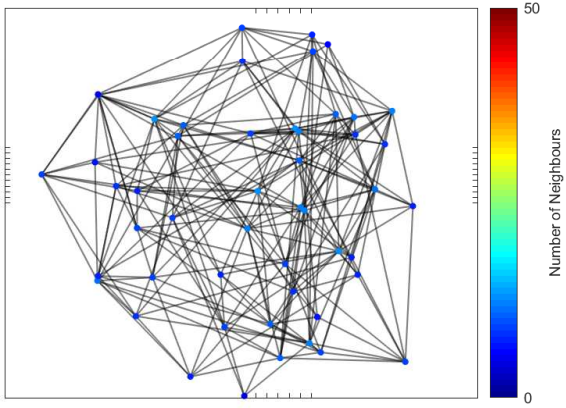}
	\caption{A sensor network with  50 nodes }
	\label{network2}
\end{figure}
\begin{figure}
	\centering
	\includegraphics[width=0.5\textwidth]{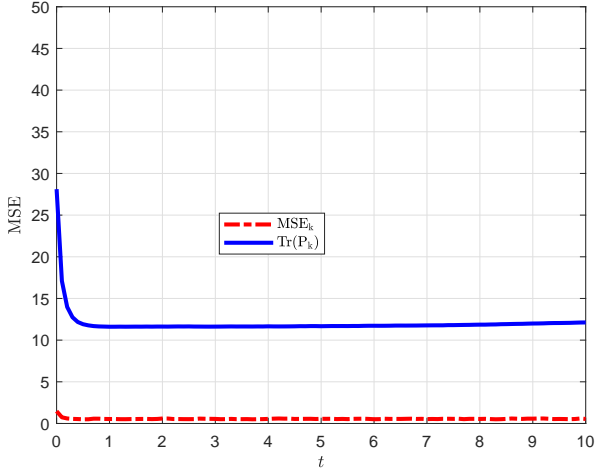}
	\caption{The consistency of DRKF}
	\label{consistency2}
\end{figure}
\begin{figure}
	\centering
	\includegraphics[width=0.5\textwidth]{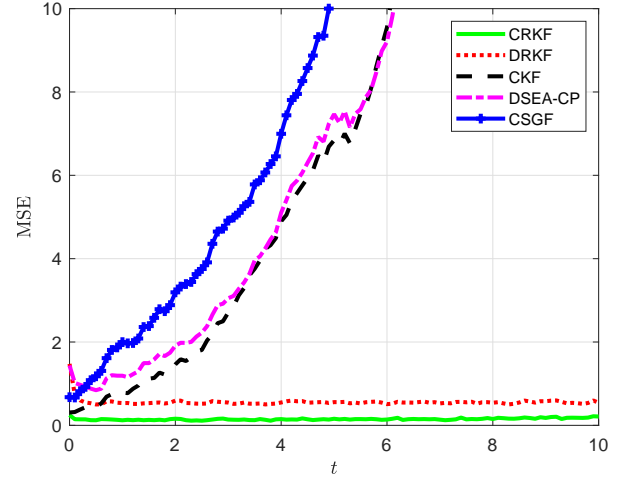}
	\caption{Comparison between  DRKF with filters from the literature}
	\label{compare2}
\end{figure}

\section{Conclusion}
This paper  studied a distributed robust state estimation problem for a class of discrete-time stochastic  systems with multiplicative noise and degraded measurements over corrupted communication channels.   Employing  local imprecise statistics, we first proposed a three-step DRKF. Then, under some mild conditions, we proved  that its  MSE is uniformly upper bounded by a constant matrix after a finite transient time. The finite time interval is related to the collective observability and the network size. { A switching fusion scheme based on a sliding-window fusion method was proposed as a DRKF-SWF algorithm to obtain a smaller upper bound of the MSE. By considering extended computational ability of the sensors, the  DRKF-SWF shows that better performance can be achieved.}

%

\appendix

\subsection{Proof of Lemma \ref{prop_1}}\label{app_lem_moment}
	We use an inductive method to prove this lemma.
	At the initial time, it follows from Assumption \ref{ass_noise} that $E\{x_{0}x_{0}^T\}\leq P_0=\varPi_{0}$. Suppose at time  $k$ that
	$E\{x_{k}x_{k}^T\}\leq \varPi_{k},\forall k\geq 0$. According to  (\ref{system_all}), $x_{k}$ is adapted to $\mathcal{F}_{k-1}$. By Assumption \ref{ass_noise}, we have 
	 $E\{\epsilon_kx_{k}\}=0$,   and $E\{w_{k}\epsilon_k\}=0.$
	For $E\{\epsilon_k^2x_{k}x_{k}^T\}$, it holds that $E\{\epsilon_k^2x_{k}x_{k}^T\}=E\{\epsilon_k^2\}E\{x_{k}x_{k}^T\}$, then 
	\begin{align*}
	&E\{x_{k+1}x_{k+1}^T\}\\[0.3em]
	&=E\{(A_{k}+F_{k}\epsilon_k)x_{k}x_{k}^T(A_{k}+F_{k}\epsilon_k)^T\}+E\{w_{k}w_{k}^T\}&&\\[0.3em]
	&\quad+E\{(A_{k}+F_{k}\epsilon_k)x_{k}w_{k}^T\}+E\{w_{k}x_{k}^T(A_{k}+F_{k}\epsilon_k)^T\}&&\\[0.3em]
	&\leq A_{k}E\{x_{k}x_{k}^T\}A_{k}^T+E\{\epsilon_k^2\} F_{k}E\{x_{k}x_{k}^T\}F_{k}^T+ E\{w_{k}w_{k}^T\}&&\\[0.3em]
	&\leq A_{k}\varPi_{k}A_{k}^T+\mu_{k} F_{k}\varPi_{k}F_{k}^T+ Q_{k}	=\varPi_{k+1}.&&
	\end{align*}
	Hence, we obtain  $E\{x_{k+1}x_{k+1}^T\}\leq \varPi_{k+1}$.

\subsection{Proof of Lemma \ref{lem_consistent}}\label{app_lem_consis}
{Regarding the filtering structure in \eqref{filter_stru}, for proof convenience,
	we denote the state estimation errors  by $\bar e_{k,i}=\bar x_{k,i}-x_{k},\tilde e_{k,i}=\tilde x_{k,i}-x_{k},\bar{\tilde e}_{k ,i,j}=\hat{\tilde x}_{k,i,j}-x_{k}$, and $e_{k,i}=\hat x_{k,i}-x_{k}$, respectively.   
	Then it is straightforward to obtain the dynamics of these estimation errors as follows
	\begin{equation}\label{error_notations}
	\begin{split}
	\bar e_{k,i}&=A_{k-1}e_{k-1,i}-w_{k-1}-\epsilon_{k-1} F_{k-1}x_{k-1}\\[0.3em]
	\tilde e_{k,i}&=(I_n-\tau_{k,i}K_{k,i}C_{k,i})\bar e_{k,i}\\[0.3em]
	&\qquad\qquad+K_{k,i}(v_{k,i}+(\gamma_{k,i}-\tau_{k,i})C_{k,i}x_{k})\\[0.3em]
	\bar{\tilde e}_{k ,i,j}&=\tilde e_{k,j}+\varepsilon_{k,i,j},j\in \mathcal{N}_{i}\\[0.3em]
	e_{k,i}&=\sum_{j\in \mathcal{N}_{i}}W_{k,i,j}\bar{\tilde e}_{k ,i,j}.
	\end{split}
	\end{equation}

	First, we make a conjecture that if $\{\hat x_{t-1,i},P_{t-1,i}\},t\geq 1$ is conditionally consistent given the channel noise $\sigma$-algebra $\mathcal{W}_{t-1}$, i.e.,  $E\{e_{t-1,i}e_{t-1,i}^T|\mathcal{W}_{t-1}\}\leq P_{t-1,i}$, then the pairs $\{\bar x_{t,i},\bar P_{t,i}\}, \{\tilde x_{t,i}, \tilde P_{t,i}\}, \{\hat x_{t,i},P_{t,i}\}$ are all conditionally consistent given $\mathcal{W}_{k}$.	In the following, we prove the conjecture.}

 Suppose at  time  $t=k$, the pair $\{\hat x_{k-1,i},P_{k-1,i}\}$, $k\geq 1$, is conditionally consistent given $\mathcal{W}_{k-1}$. 	According to Assumption \ref{ass_noise} and  \eqref{error_notations}, we have  $E\{\epsilon_{k-1}e_{k-1,i}x_{k-1}|\mathcal{W}_{k}\}=0$	
 and $E\{w_{k-1}e_{k-1,i}|\mathcal{W}_{k}\}=0$. It follows that
	\begin{align}
	&E\{\bar e_{k,i}\bar e_{k,i}^T|\mathcal{W}_{k}\}\nonumber\\[0.3em]
&\leq A_{k-1}E\{e_{k-1,i}e_{k-1,i}^T|\mathcal{W}_{k-1}\}A_{k-1}^T+Q_{k-1}\nonumber\\[0.3em]
&	\quad+\mu_{k-1} F_{k-1}E\{x_{k-1}x_{k-1}^T\}F_{k-1}^T\\[0.3em]
&\leq A_{k-1}P_{k-1,i}A_{k-1}^T+\mu_{k-1} F_{k-1}\varPi_{k-1}F_{k-1}^T+Q_{k-1}=\bar P_{k,i}.\nonumber
	\end{align}	
	In the  measurement update, according to \eqref{error_notations}, we have $\tilde e_{k,i}=(I_n-\tau_{k,i}K_{k,i}C_{k,i})\bar e_{k,i}+K_{k,i}v_{k,i}+(\gamma_{k,i}-\tau_{k,i})K_{k,i}C_{k,i}x_{k}$. By Assumption \ref{ass_noise}, $E\{\bar e_{k,i}\gamma_{k,i}|\mathcal{W}_{k}\}=0$ and $E\{\bar e_{k,i}v_{k,i}^T|\mathcal{W}_{k}\}=0$. 
	Since $v_{k,i}$ and $\gamma_{k,i}$ are  mutually independent and  $K_{k,i}$ is adapted to $\mathcal{W}_{k}$, we have
	\begin{align}\label{eq_P3}
	&E\{\tilde e_{k,i}\tilde e_{k,i}^T|\mathcal{W}_{k}\}&&\nonumber\\[0.3em]
	&\leq  (I_n-\tau_{k,i}K_{k,i}C_{k,i})E\{\bar e_{k,i}\bar e_{k,i}^T|\mathcal{W}_{k}\}(I_n-\tau_{k,i}K_{k,i}C_{k,i})^T&&\nonumber\\[0.3em]
	&\quad+\varphi_{k,i}K_{k,i}C_{k,i}\varPi_{k}C_{k,i}^TK_{k,i}^T+K_{k,i}R_{k,i}K_{k,i}^T&&\nonumber\\[0.3em]
	&\leq (I_n-\tau_{k,i}K_{k,i}C_{k,i})\bar P_{k,i}(I_n-\tau_{k,i}K_{k,i}C_{k,i})^T&&\nonumber\\[0.3em]
	&\quad +K_{k,i}\bigg(\varphi_{k,i}C_{k,i}\varPi_{k}C_{k,i}^T+R_{k,i}\bigg)K_{k,i}^T=\tilde{P}_{k,i}.
	\end{align}	
	Note that the communication channels are imperfect and the messages received by each sensor are polluted by the channel noise through (\ref{eq_corrupt1}).
		According to  Assumption \ref{ass_noise} and \eqref{error_notations},		
		we have
	\begin{align*}
	&E\{\bar{\tilde e}_{k ,i,j}\bar{\tilde e}_{k ,i,j}^T|\mathcal{W}_{k}\}&&\nonumber\\[0.3em]
	&=E\{(\tilde x_{k,j}+\varepsilon_{k,i,j}-x_{k})(\tilde x_{k,j}+\varepsilon_{k,i,j}-x_{k})^T|\mathcal{W}_{k}\}&&\nonumber\\[0.3em]
	&\leq E\{(\tilde x_{k,j}-x_{k})(\tilde x_{k,j}-x_{k})^T|\mathcal{W}_{k}\}+E\{\varepsilon_{k,i,j}\varepsilon_{k,i,j}^T|\mathcal{W}_{k}\}&&\nonumber\\[0.3em]
	&\leq \tilde P_{k,j}+\sup\{\varepsilon_{k,i,j}\varepsilon_{k,i,j}^T\}&&\nonumber\\[0.3em]
	&\leq \tilde P_{k,j}+\varUpsilon_{i,j}&&\nonumber\\[0.3em]
		&\leq\tilde P_{k,j}+\mathcal{D}_{k,i,j}+\mathcal{D}_{i,j}+\varUpsilon_{i,j}=\bar{\tilde P}_{k,i,j}+\mathcal{D}_{i,j}+\varUpsilon_{i,j},&&
	\end{align*}	
	where $\bar{\tilde P}_{k,i,j}$ is the received matrix by sensor $i$ from sensor $j$.
	In the local fusion step, $e_{k,i}=\sum_{j\in \mathcal{N}_{i}}W_{k,i,j}\bar{\tilde e}_{k ,i,j}.$
	Given $W_{k,i,j}$ in (\ref{eq_W}), according to  (\ref{eq_P3})  and the consistent estimation of the CI method \cite{Niehsen2002Information}, we have $	E\{e_{k,i}e_{k,i}^T|\mathcal{W}_{k}\}\leq P_{k,i}.$
	
	Thus, the conjecture holds. Then the conclusion is obtained based on the conjecture and the initial estimation condition in Assumption \ref{ass_noise}.

\subsection{Proof of Lemma \ref{lem_K}}\label{app_lem_K}
According to Lemma \ref{lem_consistent},  we have
	\begin{align}\label{pf_P}
	\tilde{P}_{k,i}&= (I_n-\tau_{k,i}K_{k,i}C_{k,i})\bar P_{k,i}(I_n-\tau_{k,i}K_{k,i}C_{k,i})^T&&\nonumber\\[0.3em]
	&\quad+K_{k,i}\bigg(\varphi_{k,i}C_{k,i}\varPi_{k}C_{k,i}^T+R_{k,i}\bigg)K_{k,i}^T&&\nonumber\\[0.3em]
	&=\bar P_{k,i}-\tau_{k,i}K_{k,i}C_{k,i}\bar P_{k,i}-\tau_{k,i}\bar P_{k,i}C_{k,i}^TK_{k,i}^T&&\\[0.3em]
	&\quad+\tau_{k,i}^2K_{k,i}C_{k,i}\bar P_{k,i}C_{k,i}^TK_{k,i}^T&&\nonumber\\[0.3em]
	&\quad+K_{k,i}\bigg(\varphi_{k,i}C_{k,i}\varPi_{k}C_{k,i}^T+R_{k,i}\bigg)K_{k,i}^T&&\nonumber\\[0.3em]
	&=\bar P_{k,i}-\tau_{k,i}K_{k,i}C_{k,i}\bar P_{k,i}-\tau_{k,i}\bar P_{k,i}C_{k,i}^TK_{k,i}^T&&\nonumber\\[0.3em]
	&\quad+K_{k,i}\Xi_{k,i}K_{k,i}^T&&\nonumber\\[0.3em]
	&=(K_{k,i}-K_{k,i}^*)\Xi_{k,i}(K_{k,i}-K_{k,i}^*)^T&&\nonumber\\[0.3em]
	&\quad+(I-\tau_{k,i}K_{k,i}^*C_{k,i})\bar P_{k,i},&&\nonumber
	\end{align}
	where $K_{k,i}^*=\tau_{k,i}\bar P_{k,i}C_{k,i}^T\Xi_{k,i}^{-1}$ and $\Xi_{k,i}=\tau_{k,i}^2C_{k,i}\bar P_{k,i}C_{k,i}^T+R_{k,i}+\varphi_{k,i}C_{k,i}\varPi_{k}C_{k,i}^T.$ 	
	Thus,  (\ref{pf_P}) shows that $\Tr(\tilde{P}_{k,i})$ is minimized  when
	$K_{k,i}=K_{k,i}^*=\tau_{k,i}\bar P_{k,i}C_{k,i}^T\Xi_{k,i}^{-1}.$
	As a result, $	\tilde P_{k,i}=(I-\tau_{k,i}K_{k,i}C_{k,i})\bar P_{k,i}.$ Since $K_{k,i}^*$ is a measurable function of $\bar P_{k,i}$, which is adapted to $\mathcal{W}_{k}$, also, $K_{k,i}^*$ is adapted to $\mathcal{W}_{k}$.

\subsection{Proof of Lemma \ref{lem_sufficient_obser}}\label{app_lem_obser}
	According to Lemma \ref{prop_1}, we have $\varPi_{k+1}=A_{k}\varPi_{k}A_{k}^T+\mu_{k} F_{k}\varPi_{k}F_{k}^T+ Q_{k}.$ 
Taking the 2-norm of both sides yields $\|\varPi_{k+1}\|_2\leq\|\varPi_{k}\|_2\left(\|A_{k}\|_2^2+\mu_{k}\|F_{k}\|_2^2 \right) + \|Q_{k}\|_2.$
Denote $\|A_{k}\|_2^2+\mu_{k}\|F_{k}\|_2^2=:\bar \alpha_{k}$ and $\|Q_{k}\|_2=:\bar q_{k}$. Then, $\|\varPi_{k+1}\|_2\leq\varpi_{k+1},$
where $\varpi_{k+1}=\|P_{0}\|_2\prod_{i=0}^{k}\bar\alpha_{i} + \sum_{s=1}^{k}\left(\bar q_{s-1}\prod_{j=s}^{k}\bar\alpha_{j} \right)+\bar q_{k}.$
It follows that $\bar R_{k,i}=:R_{k,i}+\varphi_{k,i}C_{k,i}\varPi_{k}C_{k,i}^T
\leq R_{k,i}+\varpi_{k}\varphi_{k,i}C_{k,i}C_{k,i}^T=\tilde R_{k,i}.$
If (\ref{Observability_matrix2}) is satisfied, (\ref{Observability_matrix}) holds.

\subsection{Proof of Lemma \ref{lem_multi_bounds}}\label{app_lem_sufficient}
	According to Lemma \ref{prop_1} and   Assumption \ref{ass_structure}, we have $\varPi_{k_{t+1}}=\Phi_{k_{t+1},k_{t}}\varPi_{k_{t}}\Phi_{k_{t+1},k_{t}}^T+ \mathcal Q_{k_{t+1},k_{t}}
	+\mu_{k_{t}} \Phi_{k_{t+1},k_{t}}F_{k_{t}}\varPi_{k_{t}}F_{k_{t}}^T\Phi_{k_{t+1},k_{t}}^T.$
	Multiplying from left   by $\mu_{k_{t+1}}F_{k_{t+1}}$ and from right  by $F_{k_{t+1}}^T$ yields
	\begin{align*}
	&\mu_{k_{t+1}}F_{k_{t+1}}\varPi_{k_{t+1}}F_{k_{t+1}}^T&&\\[0.3em]
	&=\mu_{k_{t+1}}F_{k_{t+1}}\Phi_{k_{t+1},k_{t}}\varPi_{k_{t}}\Phi_{k_{t+1},k_{t}}^TF_{k_{t+1}}^T&&\\[0.3em]
	&\quad+\mu_{k_{t+1}}F_{k_{t+1}}\mu_{k_{t}} \Phi_{k_{t+1},k_{t}}F_{k_{t}}\varPi_{k_{t}}F_{k_{t}}^T\Phi_{k_{t+1},k_{t}}^TF_{k_{t+1}}^T&&\\[0.3em]
	&\quad+ \mu_{k_{t+1}}F_{k_{t+1}}\mathcal Q_{k_{t+1},k_{t}}F_{k_{t+1}}^T,&&
	\end{align*}
	where $\mathcal Q_{k_{t+1},k_{t}}=\sum_{k=k_{t}}^{k_{t+1}}\Phi_{k_{t+1},k}Q_{k}\Phi_{k_{t+1},k}^T.$
	Denote $\mu_{k_{t}} F_{k_{t}}\varPi_{k_{t}}F_{k_{t}}^T=: \Theta_{k_{t}}$, then we have
	\begin{align}\label{eq_delta}
	&\Theta_{k_{t+1}}&&\nonumber\\[0.3em]
	&=\frac{\mu_{k_{t+1}}}{\mu_{k_{t}}}F_{k_{t+1}}\Phi_{k_{t+1},k_{t}}F_{k_{t}}^{-1}\Theta_{k_{t}}F_{k_{t}}^{-T}\Phi_{k_{t+1},k_{t}}^TF_{k_{t+1}}^T&&\nonumber\\[0.3em]
	&\quad+\mu_{k_{t+1}}F_{k_{t+1}}\Phi_{k_{t+1},k_{t}}\Theta_{k_{t}}\Phi_{k_{t+1},k_{t}}^TF_{k_{t+1}}^T&&\nonumber\\[0.3em]
&\quad 	+ \mu_{k_{t+1}}F_{k_{t+1}}\mathcal Q_{k_{t+1},k_{t}}F_{k_{t+1}}^T.&&
	\end{align}	
	Taking 2-norm of both sides of (\ref{eq_delta}) yields
	\begin{align}\label{eq_delta2}
	&\|\Theta_{k_{t+1}}\|_2&&\nonumber\\[0.3em]
&	\leq  \|\frac{\mu_{k_{t+1}}}{\mu_{k_{t}}}F_{k_{t+1}}\Phi_{k_{t+1},k_{t}}F_{k_{t}}^{-1}\Theta_{k_{t}}F_{k_{t}}^{-T}\Phi_{k_{t+1},k_{t}}^TF_{k_{t+1}}^T\|_2&&\nonumber\\[0.3em]
	&\quad+\|\mu_{k_{t+1}}F_{k_{t+1}}\Phi_{k_{t+1},k_{t}}\Theta_{k_{t}}\Phi_{k_{t+1},k_{t}}^TF_{k_{t+1}}^T\|_2&&\nonumber\\[0.3em]
	&\quad+ \mu_{k_{t+1}}\|F_{k_{t+1}}\mathcal Q_{k_{t+1},k_{t}}F_{k_{t+1}}^T\|_2&&\nonumber\\[0.3em]
&	\leq  \rho_{k_{t}} \|\Theta_{k_{t}}\|_2+ \mu_{k_{t+1}}\|F_{k_{t+1}}\mathcal Q_{k_{t+1},k_{t}}F_{k_{t+1}}^T\|_2.&&
	\end{align}
According to \cite{Elaydi2005An},  conditions (\ref{eq_exp}) and (\ref{eq_Q_sup}) now give $\sup\limits_{k_{t}\in\mathbb{N}}\|\Theta_{k_{t}}\|_2< \infty$, i.e., $\Theta_{k}$ is uniformly upper bounded.
\subsection{Proof of Theorem \ref{theorem_boundedness}}\label{app_thm_bound}
	Introduce 
	\begin{align*}
	S_{k,i}&:= P_{k,i}^{-1}\\[0.3em]
	\tilde Q_{k}&:= \mu_{k} F_{k}\varPi_{k}F_{k}^T+ Q_{k}\\[0.3em]
	G_{k,i}&:= \sum_{j\in \mathcal{N}_{i}}a_{i,j}\bar C_{k,j}^T\bar R_{k,j}^{-1}\bar C_{k,j}  \\[0.3em]
	\bar R_{k,i}&:= R_{k,i}+\varphi_{k,i}C_{k,i}\varPi_{k}C_{k,i}^T.
	\end{align*} 
	By Assumption \ref{ass_noise},
	\begin{align*}
	&\bar{\tilde P}_{k,i,j}+\mathcal{D}_{i,j}+\varUpsilon_{i,j}&&\\[0.3em]
	&=\tilde P_{k,j}+\mathcal{D}_{k,i,j}+\mathcal{D}_{i,j}+\varUpsilon_{i,j}&&\\[0.3em]
	&\geq\tilde P_{k,j}+\varUpsilon_{i,j}\geq \tilde P_{k,j}.&&
	\end{align*}
	As $\inf_{k\in\mathbb{N}} Q_{k}>0$, and $\sup_{k\in\mathbb{N}} \left[\tau_{k,i}^2 C_{k,i}^TR_{k,i}^{-1}C_{k,i}\right]<\infty$ in Assumption \ref{ass_noise}, there exists a scalar $\vartheta_0>0$  such that $\bar{\tilde P}_{k,i,j}+\mathcal{D}_{i,j}+\varUpsilon_{i,j}\leq (1+\vartheta_0)\tilde P_{k,j}$.

	According to Algorithm  \ref{alg:A} and Lemma \ref{lem_multi_bounds}, 
	\begin{align}\label{proof_S1}
	S_{k,i}=&\sum_{j\in \mathcal{N}_{i}}a_{i,j}(\bar{\tilde P}_{k,i,j}+\mathcal{D}_{i,j}+\varUpsilon_{i,j})^{-1}\\
	\geq	&\sum_{j\in \mathcal{N}_{i}}\frac{a_{i,j}}{1+\vartheta_0}\big(A_{k-1}S_{k-1,j}^{-1}A_{k-1}^T+\tilde Q_{k-1}\big)^{-1}+\frac{G_{k,i}}{1+\vartheta_0} \nonumber\\
	\geq&  \bar \eta A_{k-1}^{-T}(\sum_{j\in \mathcal{N}_{i}}a_{i,j}S_{ k-1,j}) A_{k-1}^{-1}+\frac{G_{k,i}}{1+\vartheta_0},\nonumber
	\end{align}
	where 	$0< \bar \eta<1$. This inequality  is obtained by Lemma 1 in \cite{Battistelli2014Kullback} using Assumption \ref{ass_structure} and $\frac{1}{1+\vartheta_0}<1.$ Let $a_{ij, k}$ be the $(i,j)$th element of $\mathcal{A}^k$.
	By recursively applying  (\ref{proof_S1})  $k\geq N+\bar N$ times, we have
	\begin{align}\label{proof_stability3}
	S_{k,i} 		\geq&\bar \eta^{ k}\Phi_{k,0}^{-T}\left[\sum_{j\in \mathcal{V}} a_{ij, k} S_{ 0,j}\right] \Phi_{k,0}^{-1}+\frac{	\bar S_{k,i} }{1+\vartheta_0}
	,
	\end{align}
	where
	\begin{align*}
	\bar S_{k,i}			=&
	\sum_{s=1}^{k}\bar \eta^{s-1}\Phi_{k,k-s+1}^{-T}\Bigg[\sum_{j\in \mathcal{V}} a_{ij,s}\tilde S_{k-s+1,j}\Bigg]\Phi_{k,k-s+1}^{-1},	
	\end{align*}	
	with $\tilde S_{k,j}=\bar C_{k,j}^T\bar R_{k,j}^{-1}\bar C_{k,j}.$
	%
	Since the first term of the right-hand side of (\ref{proof_stability3}) is positive definite, it follows that
	\begin{align}\label{eq_S_ineq}
	S_{k,i}\geq\frac{	\bar S_{k,i} }{1+\vartheta_0},\forall k\geq N+\bar N.
	\end{align}
	Since $\mathcal{G}$ is strongly connected,  $a_{ij,s}>0$ for $s\geq N-1$ \cite{He2017Consistent}.
	Supposing $\bar L=N+\bar N$, we obtain
	\begin{align}\label{proof_stability4}
	\bar S_{k,i}\geq &\sum_{s=1}^{\bar L}\bar \eta^{s-1}\Phi_{k,k-s+1}^{-T}
	\Bigg[\sum_{j\in \mathcal{V}} a_{ij,s}\tilde S_{k-s+1,j}\Bigg] \Phi_{k,k-s+1}^{-1}\nonumber\\
	\geq &a_{\min}\bar \eta^{ \bar L-1}\sum_{s=N}^{\bar L}\Phi_{k,k-s+1}^{-T}
	\Bigg[
	\sum_{j\in \mathcal{V}} \tilde S_{k-s+1,j}\Bigg] \Phi_{k,k-s+1}^{-1}\nonumber\\
	=&a_{\min}\bar \eta^{ \bar L-1}\sum_{j=1}^{N}\sum_{s=N}^{\bar L} \Phi_{k,k-s+1}^{-T}\tilde S_{k-s+1,j} \Phi_{k,k-s+1}^{-1},
	\end{align}	
	where $a_{\min}=\min_{i,j\in \mathcal{V},s\in \{N,\dots,\bar L\}}{a_{ij,s}>0}$.

	According to Assumption \ref{ass_structure}, there exists a   scalar $\beta>0$, such that $	\Phi_{k,k-\bar L+1}^{-T}\Phi_{k,k-\bar L+1}^{-1}\geq \beta I_n,\forall k\geq 0.$	
	From   Lemma \ref{lem_sufficient_obser} and $\bar L=N+\bar N$, it holds that
	\begin{align}
	&\sum_{j=1}^{N}\sum_{s=N}^{\bar L} \Phi_{k,k-s+1}^{-T}\tilde S_{k-s+1,j} \Phi_{k,k-s+1}^{-1}\nonumber\\[0.3em]
	=&\Phi_{k,k-\bar L+1}^{-T}\nonumber\\[0.3em]
	&\times\sum_{j=1}^{N}\Bigg[\sum_{s=k-\bar L+1}^{k-N+1}\Phi_{s,k-\bar L+1}^{T}\tilde S_{k-\bar L+1,j}\Phi_{s,k-\bar L+1}\Bigg]\Phi_{k,k-\bar L+1}^{-1}\nonumber\\[0.3em]
	\geq&\alpha \Phi_{k,k-\bar L+1}^{-T}\Phi_{k,k-\bar L+1}^{-1}	\geq \alpha\beta I_{n}, \forall k\geq N+\bar N.\label{eq_bounds}
	\end{align}	
	Summing up  (\ref{proof_stability4}) and (\ref{eq_bounds}) yields
	\begin{align}\label{eq_obser2}
	\bar S_{k,i}\geq& a_{\min}\bar \eta^{ \bar L-1}\alpha\beta I_{n}, \forall k\geq N+\bar N.
	\end{align}	
	Let $S_*(\alpha)=a_{\min}\bar \eta^{ \bar L-1}\alpha\beta I_{n}$. In light of (\ref{eq_S_ineq}), 
	it holds that $P_{k,i}^{-1}=S_{k,i}\geq S_*(\alpha)$, $\forall k\geq N+\bar N$. Hence, $\sup_{k\geq\bar  L}P_{k,i}\leq S_*^{-1}(\alpha)$. Since the filter is conditionally consistent, $\sup_{k\geq\bar  L}E\{(\hat x_{k,i}-x_{k})(\hat x_{k,i}-x_{k})^T|\mathcal{W}_{k}\}\leq S_*^{-1}(\alpha)$. Taking mathematical expectation of  both sides and denoting $ \eta=\frac{\bar \eta^{ 1-\bar L}}{a_{\min}}>0$,  the conclusion of the theorem holds.

\ifCLASSOPTIONcaptionsoff
  \newpage
\fi

\small
\bibliography{references_filtering2}
\bibliographystyle{ieeetr}

\end{document}